\documentclass[journal]{IEEEtran}
\usepackage{authblk}
\usepackage[letterpaper,margin=0.8in]{geometry}
\usepackage{amsmath}
\usepackage{multicol}
\usepackage{amssymb}
\usepackage{amsthm}
\usepackage{amsrefs}
\usepackage{amscd}
\usepackage{amsfonts}
\usepackage[pdftex]{graphicx}
\usepackage{pdfsync}
\usepackage{hyperref}
\usepackage{color}
\usepackage{verbatim}
\usepackage{enumerate}
\usepackage{cite}
\usepackage{caption}
\usepackage{subcaption}
\usepackage{array}
\usepackage{comment}
\usepackage{tikz}
\usepackage{graphicx}
            
\usepackage{algorithm}
\usepackage[noend]{algpseudocode}
\usepackage{float}
\usepackage{enumitem}

\makeatletter
\def\BState{\State\hskip-\ALG@thistlm}
\makeatother
\newtheorem{theorem}{Theorem}

\newtheorem{proposition}[theorem]{Proposition}
\theoremstyle{definition}

\newtheorem{example}{Example}

\newtheorem{Problem Statement}{Problem Statement}

\let\bm\boldsymbol
\hypersetup{
    colorlinks, linkcolor={blue},
    citecolor={blue}, urlcolor={blue}
}

\DeclareMathOperator*{\argmin}{arg\,min}

\definecolor{blue}{rgb}{0,0,1}
\definecolor{cof}{RGB}{219,144,71}
\definecolor{pur}{RGB}{186,146,162}
\definecolor{greeo}{RGB}{91,173,69}
\definecolor{greet}{RGB}{52,111,72}

\graphicspath{ {./Images/} }

\begin{document}
\title{Information Diffusion of Topic Propagation in Social Media}

\author[1]{Shahin Mahdizadehaghdam\thanks{smahdiz@ncsu.edu}}
\author[1]{Han Wang\thanks{hwang42@ncsu.edu}}
\author[1]{Hamid Krim\thanks{ahk@ncsu.edu}}
\author[2]{Liyi Dai\thanks{liyi.dai.civ@mail.mil}}
\affil[1]{Department of Electrical and Computer Engineering, \authorcr
 North Carolina State University, Raleigh, NC}
\affil[2]{Army Research Office, RTP, Raleigh NC}
\date{}
\maketitle

\begin{abstract}
Real-world social and/or operational networks consist of agents with associated states, whose connectivity forms complex topologies. 
This complexity is further compounded by interconnected information layers, consisting, for instance, documents/resources of the agents which mutually share topical similarities.
Our goal in this work is to predict the specific states of the agents, as their observed resources evolve in time and get updated.
The information diffusion among the agents and the publications themselves effectively result in a dynamic process which we capture by an interconnected system of networks (i.e. layered).
More specifically, we use a notion of a supra-Laplacian matrix to address such a generalized diffusion of an interconnected network starting with the classical "graph Laplacian".
The auxiliary and external input update is modeled by a  multidimensional Brownian process, yielding two contributions to the variations in the states of the agents: one that is due to the intrinsic interactions in the network system, and the other due to the external inputs or innovations.
A variation on this theme, a priori knowledge of a fraction of the agents' states is shown to lead to a Kalman predictor problem. This helps us refine the predicted states exploiting the error in estimating the states of agents. 

Three real-world datasets are used to evaluate and validate the information diffusion process in this novel layered network approach. Our results demonstrate a lower prediction error when using the interconnected network rather than the single connectivity layer between the agents.
The prediction error is further improved by using the estimated diffusion connection and by applying the Kalman approach with partial observations.
\end{abstract}
\section{Introduction}
\label{sec:introduction}
The emergence and rapid growth of the Internet have led to far greater  access to and exchange of information between machines and among people. Social media  such as Facebook, Twitter, and LinkedIn, have indeed connected billions of people across the planet, and have in addition almost trivialized information and resource exchange. As a physical process, the information flow---through its rather extensive connectivity and resulting frequent exchanges among the nodes, as connecting entities---is a pervasive diffusion. The high activity over a typical network entails active dynamics, and hence
 large variations in content.                                                                                                                                                                                                                                                                                                                                                                                                                                         
The states of agents (nodes) in such a network are consequently actively altered by the diffusion of information, according to associated dynamical parameters of the network. Throughout this work, we use the general term "agent" to refer to users in social networks, bloggers, authors, or any individual who produces information content, and is networked with others. Agents are usually/invariably interested in different topics irrespective of the structure of the network. Bloggers, for example, may write about specific topics, while different authors may only publish articles about their own research areas. 
 
The interest of a given agent may, also, evolve over time, as a result of diffusion of information from other agents or sources.
This work is, to the best of our knowledge, first in addressing this problem in its full generality thus, in particular, we propose {\em layered and interconnected  networks} to beyond a two-layer network, as has appeared in the literature \cite{gomez2013diffusion}. This in turn allows diffusion to take place over many layers (a truly {\em multi-layer diffusion network}). 
 
To proceed, we note that the diffusion phenomenon in a network and our later formulation will also be useful in predicting the future state of agents in the network. Specifically, information diffusion process on the topical states of agents is further clarified with the following fundamental working hypotheses:
\begin{figure*}[t]
\centering
	\begin{subfigure}[t]{0.31\textwidth}
		\centering
		\includegraphics[scale=0.5]
		{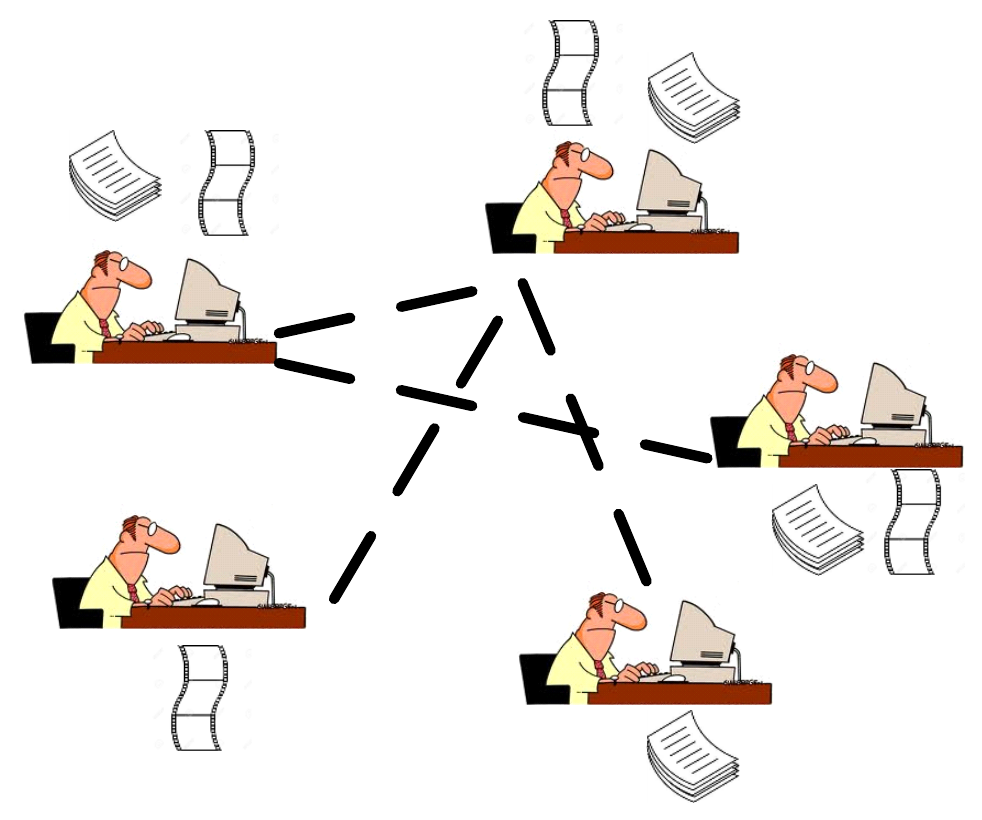}
		\caption{\footnotesize{Single layer network.}}
		\label{fig:single_layer}
	\end{subfigure}
	\begin{subfigure}[t]{0.31\textwidth}
		\centering
		\includegraphics[scale=0.8]
		{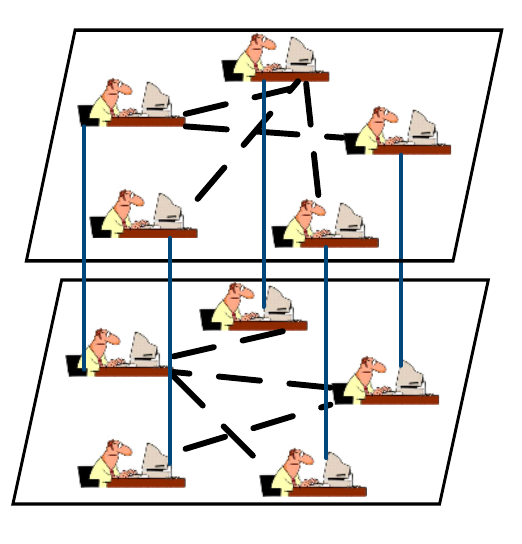}
		\caption{\footnotesize{Multiplex network.}}
		\label{fig:multi_layer}
	\end{subfigure}
	\begin{subfigure}[t]{0.31\textwidth}
    	\centering
		\includegraphics[scale=0.38]
		{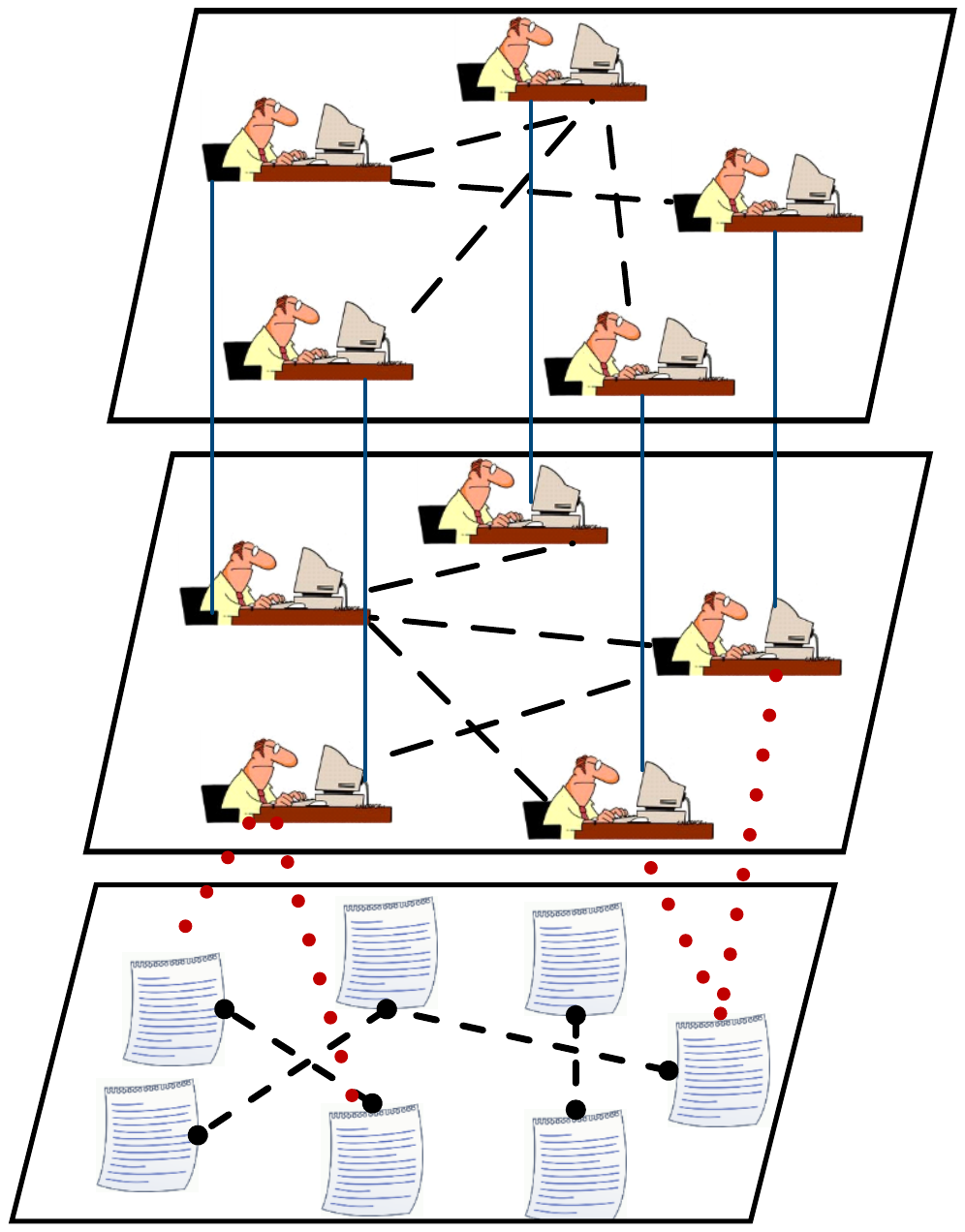}
		\caption{\footnotesize{Interconnected network of heterogeneous nodes of agents and documents.}}
		\label{fig:hetro_layer}
	\end{subfigure}
	\caption{\footnotesize{a) A single layer network between bloggers. Dashed edges show connections between  bloggers. b) A multiplex network between bloggers. The top layer is based on the hyperlink inter-connectivity of the bloggers, and the bottom layer is the friendship network between the bloggers. The straight inter-layer edges are showing that the bloggers are the same people in both layers. c) An interconnected network of heterogeneous nodes of agents and documents. The dotted edges are showing which blogger (agent) has produced which document. A similar set of dotted inter-layer edges are present from the top layer to the bottom layer which has not been depicted in the picture.}}
	\label{fig:connectivity}
\end{figure*}

 \emph{1) Information and State of  Agents:}\\ 
 To successfully model the information diffusion, it is first imperative to carefully define the type of dynamics which characterize the diffusion process. Many existing diffusion-based methods have associated characteristics such as infections, epidemics, diseases, viruses or contagions which can spread over the network. For completeness, we briefly describe some of these works:
in Susceptible-Infected-Susceptible (SIS) models \cite{SIS_M1, SIS_M2, SIS_M3} the infection can spread from any infected agent to its neighbors under certain conditions and sets of probabilities,  and the agents can  either be in susceptible or infected states. 
In this case, an agent in an infected state, will switch to a susceptible state over time, with some probability $p$. However, in Susceptible-Infected-Recovered (SIR) models \cite{SIR_M1, SIR_M2, SIR_M3, SIR_M4} the agents will switch to the recovered state, and may not catch the infection again. 
In addition, one commonly defines the state of an agent as active or a passive \cite{active_passive_1, active_passive_2}. The active agents will influence their neighbors to activate/trigger them.

Despite the subtle differences in these  models, the state space of the agents is usually binary indicating influence from a certain information.
While these types of diffusions may capture some dynamics, they present serious limitations in any practical network settings. 
A case in point is that of a diffusion process due to a blogger's read of other blogs in a network, and whose state/opinion changes as a result. This in turn, changes the dynamics of the network.  In Instagrams,  users may "repost" each other's images (with or without modification of the original image), thus implying that the  information  diffusing  through the network may come from text, image, video, and other types of documents. 




 \emph{2) The connectivity model:}\\ 
Information diffusion over a network is intimately related to the connectivity structure. While the connectivity model describes nodes with a potential  to influence/(or be influenced by) others, it may assign a quantitative value on the relation of the agents defined by  their mutual influence.
 Graphs are  most commonly used to model the connectivity among agents who are represented by nodes, and whose connection is depicted by edges.

 In the aforementioned "bloggers" example,  connectivity between bloggers highlights the fact that they read each other's blog.  Fig. (\ref{fig:single_layer}) shows a single layer connectivity structure between the bloggers; but various structures between them are also possible. We may, for instance, consider the information about the hyperlink connectivity as well as the friendship network between the bloggers. Fig. (\ref{fig:multi_layer}) depicts a two-layer connectivity structure (multiplex with two layers) between bloggers. 

 


In this work we define the state of each agent as a feature vector. Specifically, we will consider the state of each agent as a topic vector, to reflect the extent to which  an agent is associated with  specific topics. Our goal in this endeavor is to also provide the evolution of the topic vector of each agent over time. The goals and contributions in this paper are: creating a multi-layer diffusion network among agents and information contents, estimating the future state of agents, learning the structure of the supra-Laplacian matrix from previous diffusion history, and predicting the future states assuming that
a partial observation of the states of agents is available.

\begin{description} [leftmargin=0in]
  \item[1)] We go beyond the two-layer case to account for topical connectivity among the content nodes (sources of data characteristic of the information permeating the interconnected network) thus yielding a multi-layer network.  In our example, blogs consist of a set of documents generated by  bloggers over time, and  as illustrated in Fig. (\ref{fig:hetro_layer}). 
These document sets may share some topical similarities  (in spite of their distinct blogger sources). 
Structuring these relational properties into a network of document sets, introduces additional connecting edges between agents and document sets, as well as among agents. These additional paths will diffuse information at a secondary degree, specifically, an absence of a direct connection between two bloggers, does not exclude  one blogger getting updated on another blogger, and this on account of topical similarities between documents, they are associated with. We will refer to this complex structure of connected networks, as an {\em  interconnected network} of heterogeneous nodes, with advantages further discussed in  Section \ref{sec:ProblemFormalization}.
This structure will help us understand the process in which the author's topical interests are changing from a topic to another,  or why different topics are getting different amounts of attention by users in online social networks. 

\item[2)] While the current states of agents represent their engagement in different topics, they will evolve over time to reflect new documents being produced, hence updating the agents' states. 
We propose to estimate the state of the agents at time $t_1$, knowing their state at time $t_0<t_1$ with the help of the interconnected network. This is accomplished by considering a multi-layer diffusion network among the agents and documents.

\item[3)] In many cases, the actual diffusion among the agents is different from the fixed connectivity which we observe in the network (e.g, for a given friendship network between bloggers, a connection between two friend bloggers does not necessarily imply an influence of one on the other.).
By observing diffusion related patterns over a short period of time, we particularly attempt to understand such a \emph{static} influence network (reasonably for a short time interval), which has the presumably close connection to agents involvement in some common interest, with the capacity of leading to a common consensus. To that end, we learn the structure of the supra-Laplacian matrix of the underlying network of influence, which we subsequently use that in our prediction phase.

\item[4)] In certain online social networks, one often encounters a set of conservative users with strict privacy policy and a fraction of hub nodes with public information, and this clearly alters the previously designed state prediction problem.
We exploit the flexibility of the Kalman filter to address the problem, assuming that a partial observation of the states of agents is available. The Kalman filter enables us to refine the predicted states by exploiting the prediction error in the available states.
\end{description}

This paper is organized as follows: In Section \ref{sec:Related} we discuss some background and related work to our contribution. 
Section \ref{sec:ProblemFormalization} describes our basic formulation of the  problem.
We propose our new approach  in Section \ref{sec:ProposedMethod},  and present substantiating experimental  results in Section \ref{sec:Results}. We  provide some concluding remarks in Section \ref{sec:conclusion}.
\section{Related Works}
\label{sec:Related}
In recent years, a number of studies of real-networks have focused on epidemic spreading in multiplex networks.
In \cite{Wei2014JSAC}, two epidemies spreading in a two-layered multiplex network was addressed. While agents in the two-layered multiplex network are the same, their connectivity structures are different. In the state transition diagram of the model, each agent can be infected either by the first epidemy ($I_1$), or by the second epidemy ($I_2$), or possibly susceptible state ($S$). The suggested model is in the form of $S I_1 I_2 S$. The infected agent by the first epidemy should first get into a susceptible state to later get infected by the second epidemy. Also, thresholds determining the epidemies persistence, are defined. In \cite{Faryad2014PhysRevE}, the authors study a similar environment ($S I_1 S I_2 S$) in a two-layered multiplex network. The work focused on the long-term extinction, coexistence and  absolute dominance of the two epidemies with a different definition of dominance than that of  \cite{Wei2014JSAC}.
In \cite{gomez2013diffusion}, the authors discuss the diffusion dynamics in a two-layered multiplex network. The authors introduce supra-Laplacian matrix, and used a perturbation analysis to study the effects of changes in the spectral properties of single layer on the spectral properties of the whole network.
Similarly in \cite{Funk2010PhysRevE}, the authors study the spreading of two different processes in a multiplex network. Their work mostly focused on the impact of the degree correlation of the two layers on the epidemic properties of the processes.
The state space they considered, similar to many others was of binary nature (infected or not-infected) to model the spreading processes. 

We can classify the information diffusion process models into three major groups, probabilistic models, thermodynamic models, and counting models.
NETINF \cite{gomez2010inferring}, NETRATE \cite{rodriguez2011uncovering} and INFOPATH \cite{gomez2013structure} are the probabilistic models which infer the underlying diffusion network between information sources using consecutive hit times of the nodes by a specific cascade. The diffusion network was the result of solving an optimization problem which is computationally very expensive (super-exponential).
The main idea behind the thermodynamic models \cite{newman2010networks, escolano2012heat, gomez2013diffusion, sole2013spectral} is that heat will propagate from a warmer region to a colder region or gas will move from the region with higher density to the region with lower density. Modeling the information as heat or gas, we can write the rate at which information is changing in agent $i$ as: $\frac{d\psi_i}{dt}=D\sum_j \textbf{A}{(i, j)}(\psi_j-\psi_i)$. Where $\psi_i(t)$ is the state of the $i^{th}$ agent at time $t$, $D$ is the diffusion constant which is the amount of information passing from an agent to another agent in a small interval of time, and $\textbf{A}{(i, j)}$, the so-called adjacency matrix, reflects the connectivity between agents $i$ and $j$.
The counting models \cite{hethcote2000mathematics_SI, SIR, SIS} form counting processes to find the number of nodes in each group of susceptible or infected nodes.
Assuming that each agent has $\beta$ contacts with other agents per unit time, and in each contact of an infected agent with a susceptible node, the diseases will definitely spread, the overall rate of new infections is $\beta S X/n$, where $n$ is the total number of nodes. $X$ and $S$ are the number of infected and susceptible individuals respectively. We therefore, can write the rate of change of $X$ and $S$ as $\frac{dX}{dt}=\beta \frac{S X}{n}$ and $\frac{dS}{dt}= - \beta \frac{S X}{n}$.

In this paper we address following limitations from the aforementioned information diffusion models.

\emph{ A) }Almost all existing information diffusion models are based on  the contagion principle, using an abstract term or a clear data tag.  Specifically, many  information diffusion models in social networks, for example,  consider a specific "Hashtag" or repost as contagion \cite{gomez2013diffusion, sole2013spectral, escolano2012heat, gomez2010inferring}.
We would, thus, need to advance existing information diffusion models, if one were interested in, tracking the way bloggers influence  each other, or how articles and related research fields vary  over time, or parsing  through very large multi-modal data sets.

\emph{ B) }  A few information diffusion models \cite{gomez2013diffusion, sole2013spectral} have recently considered  multiplex connectivity models. The overwhelming majority of these works model the connectivity by just using simple graphs. 
Agents may, however, display several forms of connectivity, through multiple online social networks, or other social clubs, thus making   simple graph models  highly biased for capturing  the real world complexity of information diffusion models.

\emph{ C) } Existing models are layered by lack of considering external perturbations/inputs.
Specifically, considering the bloggers' example, the documentation evolves with time, and so do the nodes of the network, hence  impacting the mutual influence the agents have on each other. This is indeed key to one's ability to proceed to any future state prediction on the network.
In the present work, we model each agent state vector as a topic vector, effectively as a topic proportion associated with the agent's own documents. In so doing,  our solution predicts the users' future topic vectors by considering their interactions inside the interconnect network  agents, as well as  those of the documents. The latter network is allowed to vary over time on account of both the diffusion effect as well as that of the external input as new  documents are added. 

Our work is similar to the research in \cite{gomez2013diffusion}, we however, consider the interconnected network of heterogeneous nodes of documents and agents (Fig. \ref{fig:hetro_layer}) rather than the simpler multiplex network (Fig. \ref{fig:multi_layer}) in \cite{gomez2013diffusion}. We will also generalize the scalar state space of agents to a higher dimensional state space of topic proportions. Furthermore, we will consider an open system, and we study the external effect on each agent from the agents or documents which have not been captured by the diffusion in the network.
\section{Problem Formulation}
\label{sec:ProblemFormalization}\

We represent agents within a social network setting, by nodes on a graph, and their connectivity by edges between them. The connectivity may reflect friendship, co-authorship, or other  forms of behavioral similarities. 
If we consider a set $V_A$ of $N$ agents, we can imagine $M_A$ different graphs $G_A^{(m)}=(V_A^{(m)},E_A^{(m)}),\; m=1,...,M_A$  based on the $M_A$ different connectivity structures among these $N$ agents. 
We will use throughout, the term ``agent-layer'' to refer to the network comprised of communicating  agents. 
In Fig. (\ref{fig:hetro_ProbForm}) we can see two agent-layers. 
We note in this case, that the agents are the same in both layers, and we hence  have a direct one to one  connections between them (the blue edges in Fig. (\ref{fig:hetro_ProbForm})). 

Agents are usually assumed to possess some data or to produce information in their embedding network. If, for instance,  the agents were bloggers or social network users, they would be producing blog-documents or their online accounts, if they were  researchers, they would be turning out scientific papers. 
These sets of documents, may also mutually store topical similarities which will play a role in the overall information diffusion throughout the network. Of central interest in our work, is the evolution of the topical states of the agents, as the information diffusion takes place. 
Latent Dirichlet Analysis (LDA) \cite{Blei:2003:LDA}, Latent Semantic Indexing (LSI)  \cite{Deer:1990:LSI} and  several other topic modeling algorithms afford  us the ability to represent any text document as state vectors of length $T$. $T$ is an arbitrarily chosen number of topics (dimension) in a collection of text documents. 
Considering the documents as vectors in a $T$ dimensional vector space, we can build a similarity network between these vectors using their euclidean inter-distances, as well as a neighborhood criterion between them (e.g. $K$-Nearest Neighbors algorithm, Epsilon-Neighborhood).
We consider a set $V_I$ of $S$ documents. The subscript $I$ specifies  the information nature of the network and of its members (i.e, it pertains to documents in this case), with $M_I$ different graphs $G_I^{(m)}=(V_I^{(m)},E_I^{(m)}),\; m=M_A+1,...,M_A+M_I$, each describing  a different connectivity structure among $S$ documents (the bottom layer in Fig. (\ref{fig:hetro_ProbForm})).
Note that similarly efficient tools (e.g. ISOMAP \cite{ISOMAP:2000}, Principal Component Analysis (PCA)  \cite{PCA:1987}, Multi-Dimensional Scaling (MDS) \cite{MDS:1952}) for dimension reduction are available for multi-modal documents (images, videos, ...), to possibly represent them as separate networks.
In this work, we refer to the network among the documents as the information-layer. Table \ref{tbl:notation} lists all notations used in this paper along with their definitions. 
\begin{figure}[!t]
\centering
\includegraphics[scale=0.4]{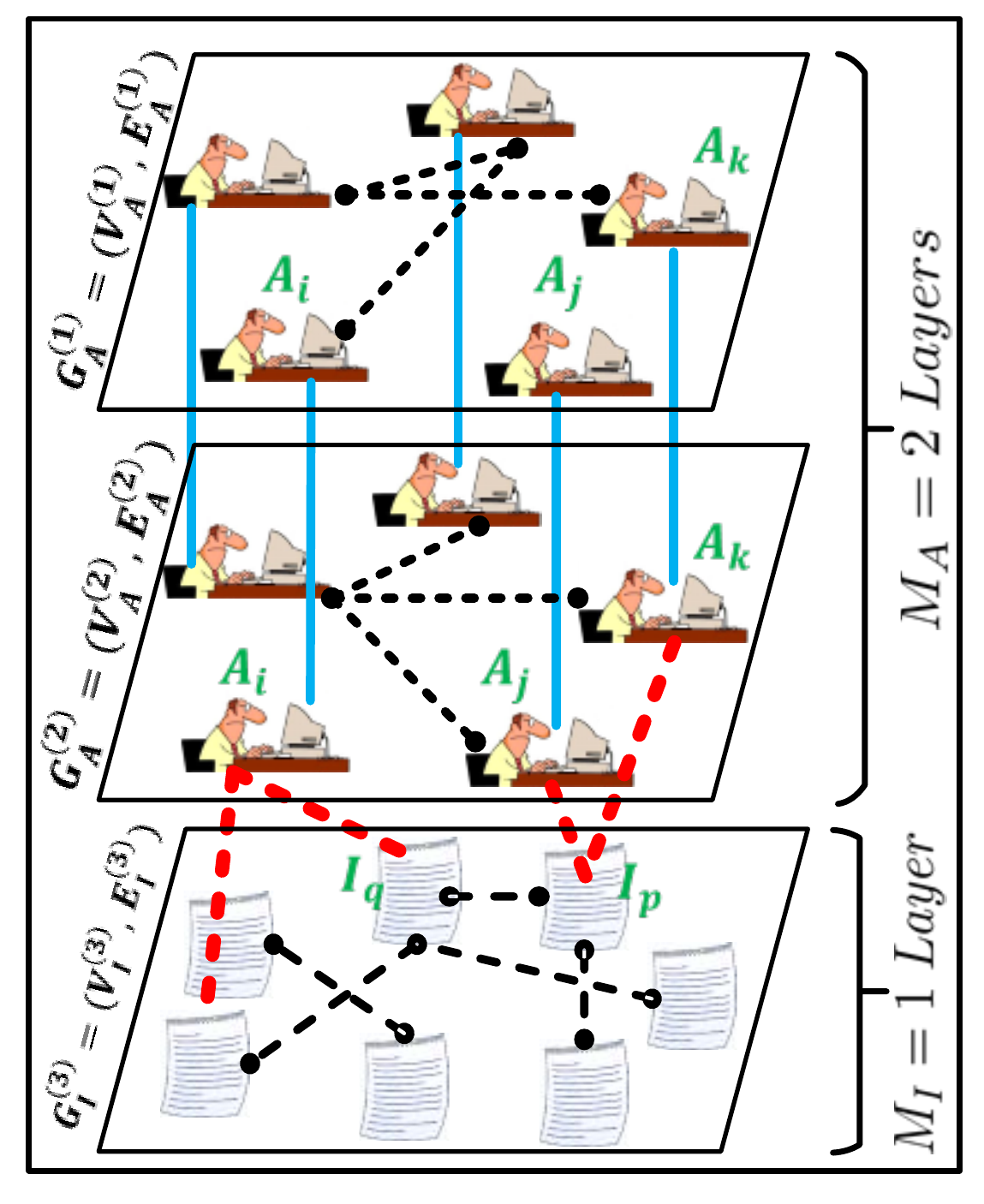}
\caption{\footnotesize{The interconnected network $\mathcal{I}$ (Eqn. (\ref{eq:interconnected_graph_I})) with two agent-layers and one information-layer. A set of inter-layer edges (similar to the red lines) are present from the top layer to the bottom layer which have not been depicted in the picture for simplicity.}}
\label{fig:hetro_ProbForm}
\end{figure}
\begin{table*}[!t]
\renewcommand{\arraystretch}{1.1}
\caption{Terminology}
\label{tbl:notation}
\begin{center}
\begin{tabular}{ | c | l || c | l | } 
  \hline
  \scriptsize{\textbf{Symbol}} & \scriptsize{\textbf{Definition}} & \scriptsize{\textbf{Symbol}} & \scriptsize{\textbf{Definition}} \\
  \hline
 \scriptsize{$A$, $I$} & \scriptsize{Agents} \& \scriptsize{documents resp.} & \scriptsize{$\textbf{W}^{(k)}_A$} & \scriptsize{Intra-layer connectivity matrix of agents in layer $k$.} \\
  \scriptsize{$A_i$, $I_p$} & \scriptsize{Agent $i$ \& documents $p$ resp.} & \scriptsize{$\textbf{W}_{A,I}^{(k,m)}$} &  \scriptsize{Inter-layer conn. of layer $k$ to $m$ among agents and docs.}\\
  \scriptsize{$G_A^{(m)}$, $G_I^{(m)}$} & \scriptsize{Intra-layer graphs of layer $m$.} & \scriptsize{$D^{(k)}_A$} & \scriptsize{Intra-layer diffusion constant of the agents in layer $k$.}\\
  \scriptsize{$\mathcal{I}$} &  \scriptsize{The interconnected graph.} & \scriptsize{$D_{A}^{(k,l)}$} & \scriptsize{Diffusion constant of agents from layer $k$ to $l$.}\\
  \scriptsize{$T$} & \scriptsize{Length of the topic vectors.} & \scriptsize{$D_{A,I}^{(k,m)}$} & \scriptsize{Diffusion const. of layer $k$ to $m$. among agents and docs.} \\ 
  \scriptsize{$\textbf{L}^{(m)}$} & \scriptsize{Laplacian matrix of layer $m$.} & \scriptsize{$N$, $S$, $P$} & \scriptsize{\# of agents, documents \& all nodes resp.}\\ 
  \scriptsize{$\textbf{x}_{A_i}^{(m)}$} &  \scriptsize{State of  agent $i$ in layer $m$.} &  \scriptsize{$\bm{\mathcal{L}}$} &  \scriptsize{Supra-Laplacian matrix.}\\
  \scriptsize{$\textbf{x}_{I_p}$} &  \scriptsize{State of  document $p$.} &  \scriptsize{$\textbf{X}$} &  \scriptsize{A $P \times T$ matrix consisted of the state of all nodes.}\\
  \hline
\end{tabular}
\end{center}
\end{table*}
\normalsize

By adopting an adapted notation of \cite{kivela2013multilayer}, we represent the document $p$ by $I_p$, and the agent $i$ by $A_i$.
The red inter-layer edges ($E_{A,I}$) in Fig. (\ref{fig:hetro_ProbForm}) relates agents to their respective documents (publisher-publication network). 
The information of document $I_p$ is denoted by a $T$ dimensional vector $\textbf{x}_{I_p} \in {\mathbb{R}^T}$. The state of agent $A_i$ in layer $m$, is defined as the average of the $T$ dimensional vectors of their associated documents:
\begin{equation}
\label{equ:Agent_initial_state}
\textbf{x}_{A_i}^{(m)}=\frac{1}{|\mathcal{N}_I(A_i)|} \sum_{p \in \mathcal{N}_I(A_i) } \textbf{x}_{I_p},
\end{equation}
where $\mathcal{N}_I(A_i)$ is the set of documents produced by agent $A_i$ or more formally, the neighbor set of  agent $A_i$ in the information-layer, which $|\mathcal{N}_I(A_i)|$ is the cardinality of that set.
The interconnected graph in Fig. (\ref{fig:hetro_ProbForm}) may thus be denoted by:
\scriptsize 
\begin{equation}
\label{eq:interconnected_graph_I}
\mathcal{I}=\bigg( \big( V_A^{(1)} \cup V_A^{(2)} \cup V_I^{(3)} \big),  \big( E_A^{(1)} \cup E_A^{(2)} \cup E_I^{(3)} \cup E_{A,A}^{(1,2)} \cup E_{A,I}^{(1,3)} \cup E_{A,I}^{(2,3)} \big) \bigg). 
\end{equation}
\normalsize

With this inter-connected formulation, and in analyzing an associated information diffusion process, we address in this paper,

\begin{itemize}
  \item The ability to achieve a diffusion of information between agents via multiple connectivity structures. The same set of agents may have different intra-layer connectivity  at  different layers;
  \item The ability to model networked documents in a separate network layer, will enable us to consider the similarities, and their evolution overtime between the documents as an information diffusion medium.
  For instance, agents $A_i$ and $A_j$ in Fig. (\ref{fig:hetro_ProbForm}) are not connected via any paths through the top two agent-layers $G_A^{(1)}$, and $G_A^{(2)}$. There is, however, a path between these two agents by way of the similarity of their documents $I_p$ and $I_q$ in $G^{(3)}_I$.
The intuition supporting the interaction between $A_i$ and the $A_j$ in the blogger's example, follows from: blogger $A_i$ noticing $I_p$ document's  similarity (written by blogger $A_j$) to his/her $I_q$ document, would read the document and being influenced, as a result.
  \item The preservation of the conventional information diffusion structures such as the co-authorship network in the introduced interconnected network model (e.g the agents $A_j$ and $A_k$ who have collaborated to produce the $I_p$). 
\end{itemize}

\setcounter{equation}{10}
\begin{figure*}[!t]
\begin{equation}
\label{eq:Sup_L}
\scriptsize
\bm{\mathcal{L}}= \left[
\begin{array}{ccc} 
 D^{(1)} \textbf{L}^{(1)}+D^{(1,2)} \textbf{I}+D^{(1,3)}\textbf{K}^{(1,3)} &   -D^{(1,2)}\textbf{I}   &   -D^{(1,3)}\textbf{W}^{(1,3)}\\
 -D^{(2,1)}\textbf{I} &   D^{(2)}\textbf{L}^{(2)}+D^{(2,1)}\textbf{I}+D^{(2,3)}\textbf{K}^{(2,3)}   &   -D^{(2,3)}\textbf{W}^{(2,3)}\\
 -D^{(3,1)}\textbf{W}^{(3,1)} &   -D^{(3,2)}\textbf{W}^{(3,2)}   &   D^{(3)}\textbf{L}^{(3)}+D^{(3,1)}\textbf{K}^{(3,1)}+D^{(3,2)}\textbf{K}^{(3,2)}
\end{array} 
\right].
\end{equation}
\hrulefill
\end{figure*}
\setcounter{equation}{2}
\normalsize 
\section{The Proposed Method}
\label{sec:ProposedMethod}

In this section we address the information diffusion across agents as a result of topic adoption and adaptation, as well as external topic additions.
To that end, we next consider a closed interconnected network with no additions. In Section  \ref{subsec:opensystem}, we consider an open interconnected network and account for an innovation injection. In Section \ref{subsec:EstimateL}, we define the estimation of the supra-Laplacian matrix using learning data. In Section \ref{subsec:Kalman}, we further refine the predicted state of the nodes using Kalman filtering. In Section \ref{subsec:eigen}, we analyze the effect of a weak inter-layer connection on the over-all diffusibility of the interconnected network.

\subsection{Closed system: Information diffusion in a heterogeneous network}
\label{subsec:closesystem}
In closed systems, all changes in the states of agents are a result of interaction of the agents in the network.
In a single layer diffusion process \cite{newman2010networks, escolano2012heat, gomez2013diffusion, sole2013spectral}, an agent state maybe formalized as:
\begin{equation}
\label{eq:change_rate}
\frac{d\textbf{x}_{A_i}}{dt}=D\sum_{j=1}^N \textbf{W}(i,j)(\textbf{x}_{A_j}-\textbf{x}_{A_i}),
\end{equation}
where $\frac{d\textbf{x}_{A_i}}{dt}$ is the $i^{th}$ agents' topic vector change over time.
$\textbf{W}(i,j)$ reflects  the connectivity status between agents $i$ and $j$, while  $D$ is the diffusion constant, or the fractional amount of information passing from $j$ to $i$ in a small time interval. By further simplification of Eqn. (\ref{eq:change_rate}),  we may write  $\frac{d \textbf{X}_{A}}{dt}$ as follows:
\begin{equation}
\label{eq:change_rate_genralize}
\frac{d \textbf{X}_{A}}{dt} + D \textbf{L} \textbf{X}_A = 0,
\end{equation}
where $\textbf{X}_{A}$ is an $N \times T$ matrix, where row $i$ is denoted by  the row vector $\textbf{x}_{A_i}^T$;
$\textbf{L}$ is an $N \times N$ graph Laplacian matrix  (i.e., $\textbf{L}=\textbf{K}-\textbf{W}$, $\textbf{K}$ being the diagonal matrix of the nodes' degree). 

Independent of the nodes in the network being agents or documents, we can state the following,
\begin{proposition}
We can generally write the supra-Laplacian matrix of an $M$ layer multiplex network with $N$ nodes in each layer as $\bm{\mathcal{L}} = \bm{\mathcal{L}}_L + \bm{\mathcal{L}}_I$. Where $\bm{\mathcal{L}}_L$ is the supra-Laplacian matrix of the intra-layer connectivity and $\bm{\mathcal{L}}_I$ is the supra-Laplacian matrix of the inter-layer connectivity. $\bm{\mathcal{L}}_L$ may be in turn, written as direct sum of the Laplacian matrices of the independent intra-layer connectivities:

\begin{equation}
\label{eq:L_Summation}
\bm{\mathcal{L}} = \bm{\mathcal{L}}_L + \bm{\mathcal{L}}_I,
\end{equation}
\footnotesize 
\begin{equation}
\label{eq:direct_sum_L}
\bm{\mathcal{L}}_L= \bigoplus_{\alpha=1}^M D^{(\alpha)} \textbf{L}^{(\alpha)} =\left[
\begin{array}{cccc} 
   D^{(1)}\textbf{L}^{(1)}&  &  & \\
     & D^{(2)}\textbf{L}^{(2)}&  & \\
  &  &\ddots&\\
  &   &   &D^{(M)}\textbf{L}^{(M)}
\end{array} 
\right].
\end{equation}
\normalsize
The inter-layer supra-Laplacian can be written as $\bm{\mathcal{L}}_I= \sum_{\alpha=1}^M (\bm{\mathcal{K}}^{\alpha}_I-\bm{\mathcal{W}}^{\alpha}_I)$, where the $\bm{\mathcal{K}}^{\alpha}_I$ is the diagonal inter-layer degree matrix of layer $\alpha$, showing the inter-layer degree of the nodes in layer $\alpha$ and the $\bm{\mathcal{W}}^{\alpha}_I$ is the inter-layer connectivity matrix of the nodes in layer $\alpha$ with the nodes in the other layers. The $\bm{\mathcal{K}}^{\alpha}_I$ and  $\bm{\mathcal{W}}^{\alpha}_I$ are formally defined in Eqns. (\ref{eq:K_I} and \ref{eq:W_I}) respectively:
\begin{equation}
\label{eq:K_I}
\bm{\mathcal{K}}^{\alpha}_I=\textbf{e}_{(\alpha, \alpha)} \otimes (\sum_{\beta=1 \beta \neq \alpha }^M D^{(\alpha, \beta)} \textbf{K}^{(\alpha, \beta)}),
\end{equation}
\begin{equation}
\label{eq:W_I}
\bm{\mathcal{W}}^{\alpha}_I= \sum_{\beta=1 \beta \neq \alpha }^M (\textbf{e}_{(\alpha, \beta)} \otimes (D^{(\alpha, \beta)} \textbf{W}^{(\alpha, \beta)})),
\end{equation}
where $\textbf{K}^{(\alpha,\beta)}$, is the diagonal matrix reflecting the degree of each node in the inter-layer
connectivity between layer $\alpha$ and layer $\beta$, $\textbf{W}^{(\alpha,\beta)}$ quantifies the inter-layer connectivity of the layer $\alpha$ nodes to the layer $\beta$ nodes and $\textbf{e}_{(\alpha, \beta)}$ is 
an all 0, $M \times M$, matrix with an only 1 element in $(\alpha, \beta)$.
\end{proposition}

\begin{proof}
Considering the interconnected network case depicted  in Fig. (\ref{fig:hetro_ProbForm}), with $M_A$ and $M_I$ layers of connectivity among the agents and the documents respectively, we can write,
\begin{align}
\label{eq:rate_X}
        \frac{d\textbf{x}_{A_i}^{(k)}}{dt}&=D_A^{(k)} \sum_{j=1}^N \textbf{W}^{(k)}_A(i,j)(\textbf{x}_{A_j}^{(k)}-\textbf{x}_{A_i}^{(k)})\\
                     &+ \sum_{l=1}^{M_A} D_{A}^{(k,l)} (\textbf{x}_{A_i}^{(l)}-\textbf{x}_{A_i}^{(k)}) \notag\\
                     &+ \sum_{m=M_A+1}^{M_A+M_I} D_{A,I}^{(k,m)} \sum_{p=1}^S \textbf{W}^{(k,m)}_{A,I} (i,p) (\textbf{x}_{I_p}^{(m)}-\textbf{x}_{A_i}^{(k)}). \notag
\end{align}
The first term of Eqn. (\ref{eq:rate_X}) represents the intra-layer diffusion of the information inside layer $k$, 
the second term accounts for the diffusion of the information between different agent-layers for agent $i$,
while the third term describes the diffusion of information from different information-layers to agent $i$ at layer $k$. 
$\textbf{W}^{(k)}_A$ depicts the connectivity of agents within layer $k$, while $\textbf{W}^{(k,m)}_{A,I}$ does that of agents in layer $k$ and documents in layer $m$. 
The $D^{(k)}_A$ is the intra-layer diffusion constant of  agents in layer $k$,
while $D^{(k,l)}_A$ is the inter-layer diffusion constant of  agents from layer $k$ to  agents in layer $l$, 
and $D^{(k,m)}_{A,I}$ is the inter-layer diffusion constant between agents in layer $k$ and  documents in layer $m$.

Further simplification of Eqn. (\ref{eq:rate_X}) yields  the following differential equation:
\begin{equation}
\frac{d\textbf{X}}{dt}= - \bm{\mathcal{L}} \textbf{X},
\label{eq:diff_X}
\end{equation}
where $\textbf{X}$ is an $P \times T$ matrix. In case of an M layer multilayer network, rewriting the $\bm{\mathcal{L}}$ matrix will lead us to Eqn. (\ref{eq:L_Summation}).
\end{proof}

Following is an example which further clears steps of the proof.

\begin{example}
 In a three-layer interconnected network similar to Fig. (\ref{fig:hetro_layer}), $\textbf{X}$ is a $P \times T$ matrix ($P=2N+S$) which represents the states of  agents in the top two layers ($\textbf{X}^{(1)}$ and $\textbf{X}^{(2)}$),  and of the topic vectors of the documents in the third layer ($\textbf{X}^{(3)}$).
Following the terminology in \cite{gomez2013diffusion}, we refer to  $\bm{\mathcal{L}}$ as a supra-Laplacian matrix, by its ability to capture the diffusion in  inter-connected network system.
Assuming undirected graphs in each layer, and symmetric diffusion constants ($D^{(k,m)}=D^{(m,k)}$), we can write for easier explanation, the supra-Laplacian matrix in case of a three layer-network (Fig. (\ref{fig:hetro_ProbForm})), with two agent-layers and one information-layer as Eqn. (\ref{eq:Sup_L}). For this specific a 3 layer inter-connected network, we simplify our notation by dropping the subscripts (as in Eqn. (\ref{eq:Sup_L}), specifying the nature of a network layer) in favor of superscripts. \footnote{Layers 1-2 are the agent-layers, while layer 3 is the information-layer.}
This hence makes matrix  $\textbf{L}^m$ represent the Laplacian of the intra-layer connectivity matrix of  layer $m$, $\textbf{I}$ the identity matrix,  while $\textbf{K}^m$ is a diagonal matrix of node degree of layer $m$, and $\textbf{K}^{m,n}$ is a diagonal matrix reflecting the degree of each node in the inter-layer connectivity between layer $m$ and layer $n$. This hence yields Eqn. (\ref{eq:Sup_L}),
which the first elements in the diagonal entries, form $\bm{\mathcal{L}}_L$ (Eqn. (\ref{eq:direct_sum_L})), the other two elements in the diagonal entries form $\bm{\mathcal{K}}_I^{\alpha}$s (Eqn. (\ref{eq:K_I})), and the rest (off-diagonal entries), form $\bm{\mathcal{W}}_I^{\alpha}$s (Eqn. (\ref{eq:W_I})).
\setcounter{equation}{11}

\hfill$\square$
\end{example}

\subsection{Open System Diffusion: Impact  of External Effects}
\label{subsec:opensystem}

Much of the existing work in information diffusion models have a limited scope (of agents, documents, parameters) when predicting the future state of the nodes. More specifically, agent states may be varied by external sources which are not captured in the network, or by some agent actions which may even to some extent, conflict with the model prediction. to address this additional auxiliary input, we propose open system model as follows:


First, consider a single layer network (agent-layer) for simplicity, we can then model the rate of change in the state of node $i$ as follows,
\begin{equation}
\label{eq:open1}
d\textbf{x}_{A_i}=D\sum_{j=1}^N \textbf{W}(i,j)(\textbf{x}_{A_j}-\textbf{x}_{A_i}) dt + \bm{\sigma}_i \; d\bm{B}(t),
\end{equation}
where  $\textbf{W}(i,j)$ quantifies the connectivity between agent $i$ and $j$, and $D$ is the diffusion constant reflecting the infinitesimal amount of information passing from $j$ to $i$ in a small interval of time, and  
$\textbf{B}(t)$ is a $T \times T$ matrix, whose columns are $T$-dimensional vectors with components as independent standard Brownian motions of variances $\bm{\sigma}_i$.
Inspired by the Ornstein-Uhlenbeck (O.U.) process \cite{Ornstein1930}, Eqn. (\ref{eq:open1}) describes the velocity of the topical-state of the nodes as a Brownian motion in presence of friction. In other words, to describe the uncertainty due to external effects, we proceed to view the whole system as a massive Brownian particle. The drift term (first term in right-hand side of  Eqn. (\ref{eq:open1})), however,  moves the velocity from a martingale state of  $\bm{\sigma}_i d\textbf{B}(t)$ towards a consensus (captured by the drift term).
In matrix form,  Eqn. (\ref{eq:open1}) may be written as,
\begin{equation}
\label{eq:open2}
d\textbf{X}_{A}(t)=-D\textbf{L}\textbf{X}_{A}(t)dt+ \bm{\Sigma}_A d\textbf{B}(t),
\end{equation}
where $\bm{\Sigma}_A$  is an $N \times T$ matrix, and each row shows the $\bm{\sigma}_i$ vector in Eqn. (\ref{eq:open1}) for agent $i$. Fig. (\ref{fig:cons}) shows numeric of examples of Eqn. (\ref{eq:open2}) with $N=5$ and $T=1$. The $\bm{\Sigma}_A$ matrix adjusts the effect of the consensus and the Brownian motion in Eqn. (\ref{eq:open2}). In this figure, we have used the term $\frac{||\bm{\Sigma}_A||_F}{||\textbf{X}_0||_F}$ (with $\textbf{X}_0=\textbf{X}_{A}(0)$) to compare the effect of $\bm{\Sigma}_A$ magnitude matrix on the states of the agents. As the value of $\frac{||\bm{\Sigma}_A||_F}{||\textbf{X}_0||_F}$ increases the second term on the right-hand side of Eqn. (\ref{eq:open2}) will have more influence on determining the state of the nodes and consequently, the state of the agents will be changed by higher uncertainty.
\begin{figure}
\centering
	\begin{subfigure}[t]{0.2\textwidth}
		\centering
		\includegraphics[scale=0.2]
		{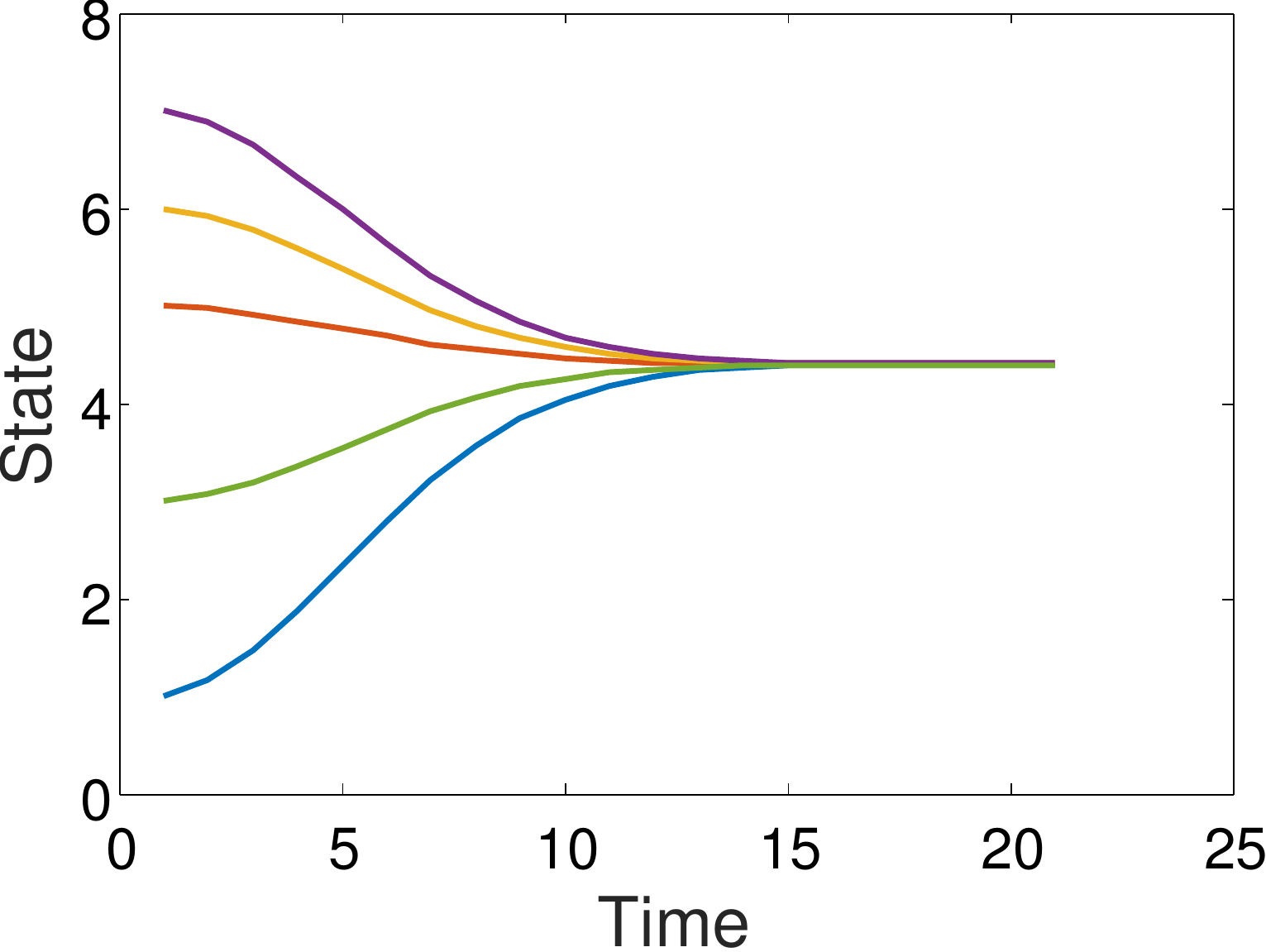}
		\caption{$\frac{||\bm{\Sigma}_A||_F}{||\textbf{X}_0||_F}=0$}
		\label{fig:cons_0}
	\end{subfigure}
	\begin{subfigure}[t]{0.2\textwidth}
		\centering
		\includegraphics[scale=0.2]
		{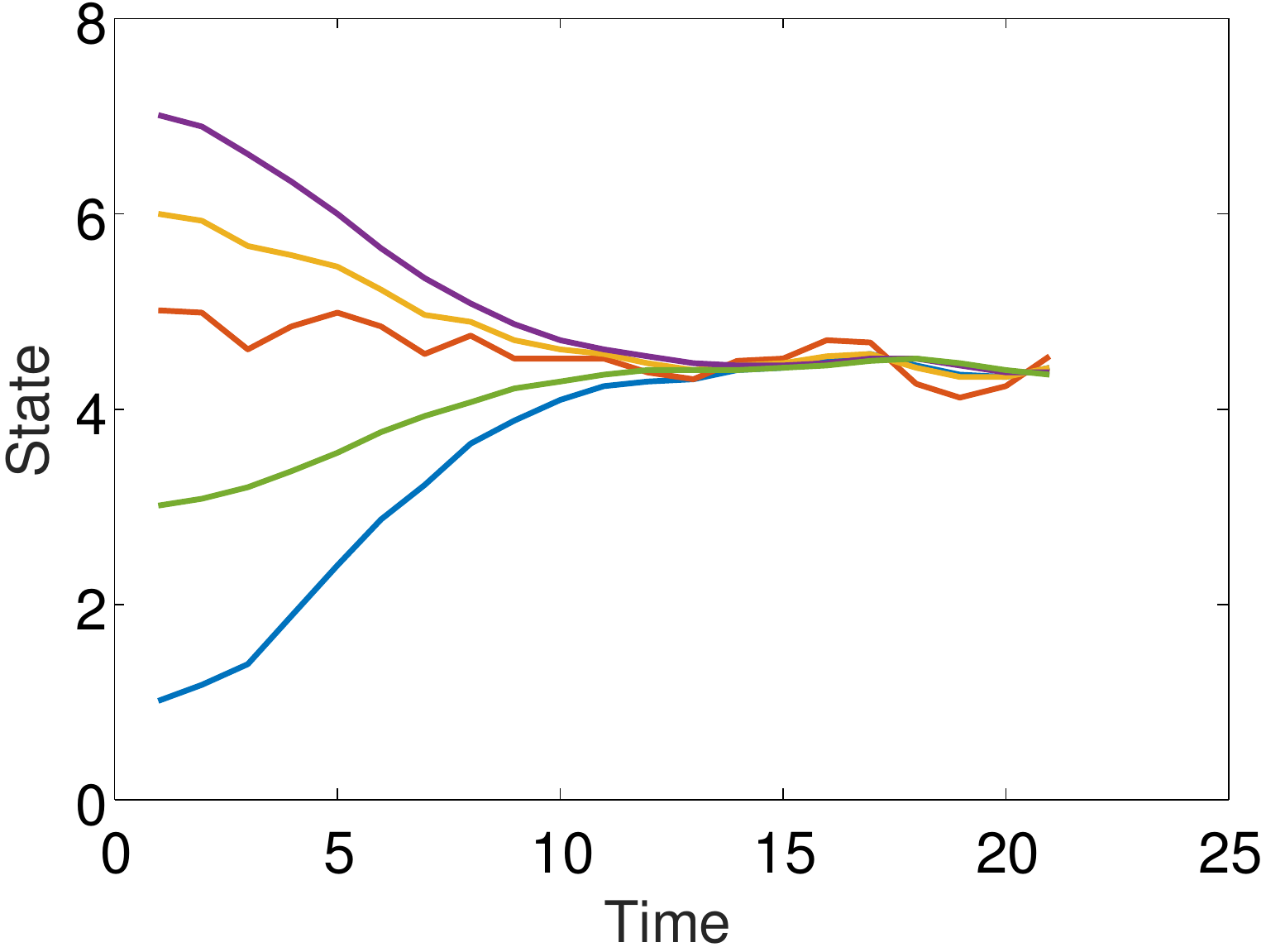}
		\caption{$\frac{||\bm{\Sigma}_A||_F}{||\textbf{X}_0||_F}=0.24$}
		\label{fig:cons_2}
	\end{subfigure}\\ \vspace{2mm}
	\begin{subfigure}[t]{0.2\textwidth}
    	\centering
		\includegraphics[scale=0.2]
		{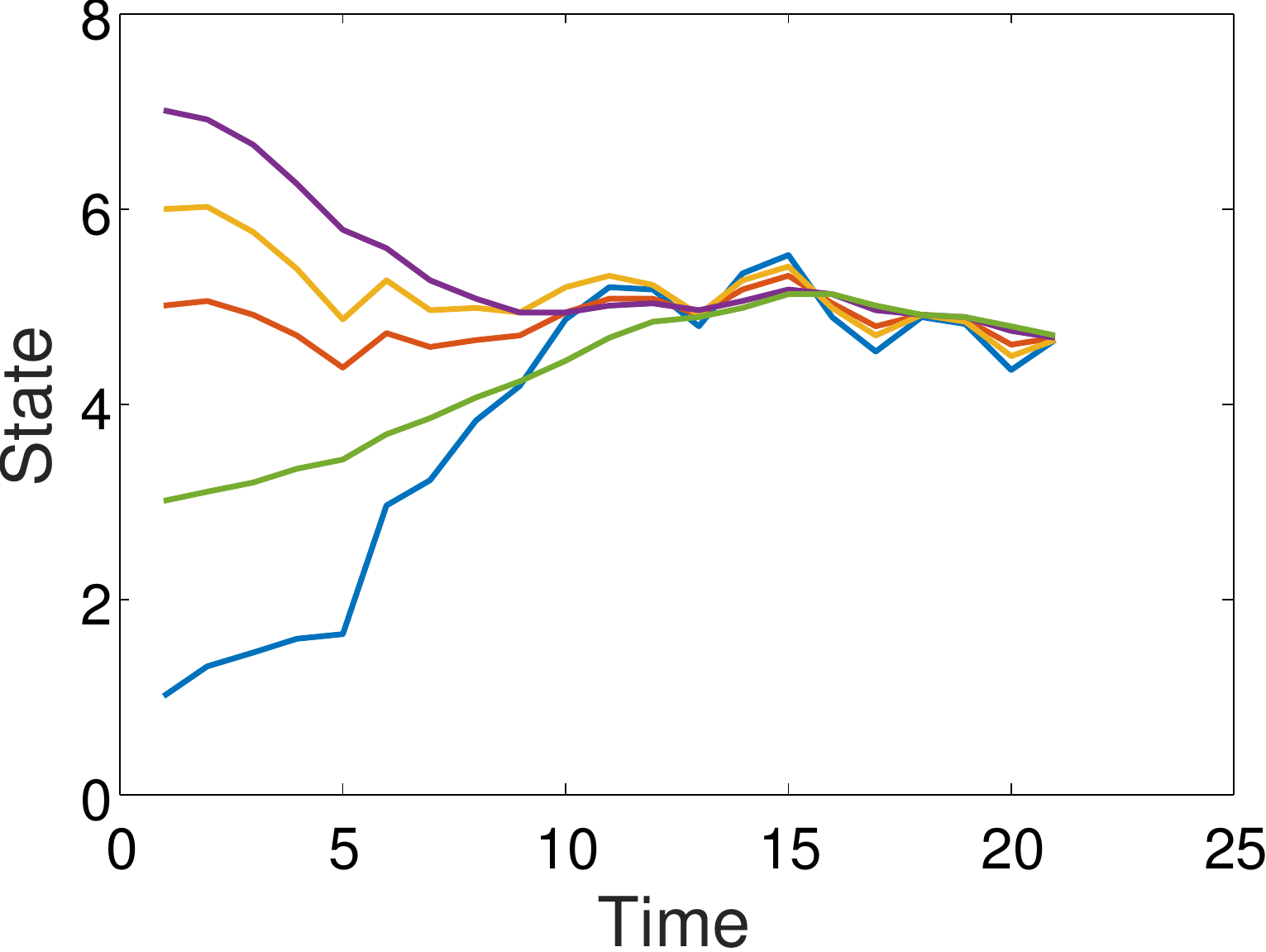}
		\caption{$\frac{||\bm{\Sigma}_A||_F}{||\textbf{X}_0||_F}=0.42$}
		\label{fig:cons_5}
	\end{subfigure}
	\begin{subfigure}[t]{0.2\textwidth}
    	\centering
		\includegraphics[scale=0.2]
		{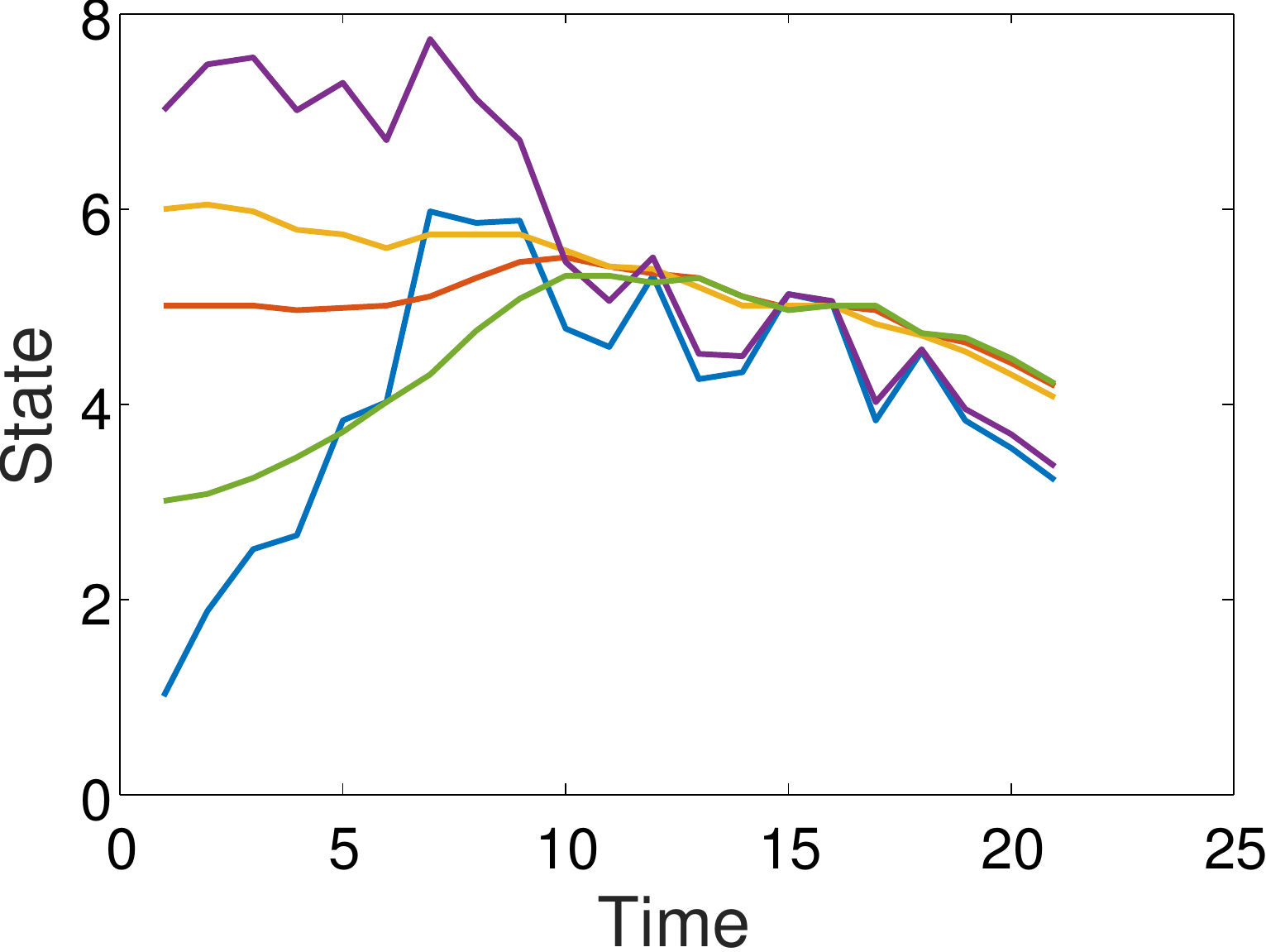}
		\caption{$\frac{||\bm{\Sigma}_A||_F}{||\textbf{X}_0||_F}=0.67$}
		\label{fig:cons_8}
	\end{subfigure}
	\caption{The effect of different magnitude of $\bm{\Sigma}_A$ matrix on the state of agents.}
	\label{fig:cons}
\end{figure}

To extend to an interconnected setting, we can state the following statement: Given the states of nodes at time $t_0$ we can predict the states at time $t_1$ , $t_1 > t_0$ as follows:
\begin{equation}
\label{eq:open4}
\widehat{\textbf{X}}(t_1)=e^{-\bm{\mathcal{L}}(t_1-t_0)} \textbf{X}(t_0)+ \int_{t_0}^{t_1} e^{\bm{\mathcal{L}}(s-t_1+t_0)} \bm{\Sigma} d\textbf{B}(s),
\end{equation}
where $e^{-\bm{\mathcal{L}}(s-t_1 + t_0)}$ is a matrix exponential and itself is a  $P \times P$ matrix and $\bm{\Sigma}$ is $P \times T$ matrix.
To this end, we rewrite Eqn. (\ref{eq:open2}) for interconnected networks as,
\begin{equation}
\label{eq:open3}
d\textbf{X}(t)=-\bm{\mathcal{L}}\textbf{X}(t)dt+\bm{\Sigma} d\textbf{B}(t).
\end{equation}
The first term on the right-hand side of  Eqn. (\ref{eq:open3}) describes the network-level diffusion taking place among the nodes (mean reverting term), while the second depicts the global diffusion process, affecting nodes regardless of their interactions.

There are multiple ways to calculate the term $\int_{t_0}^{t_1} e^{\bm{\mathcal{L}}\tau}\; d\tau$,
one may use the basic definition of a matrix exponential, and calculate the numerical value of the integral, or instead use the Jordan form as $\bm{\mathcal{L}}=\bm{\mathcal{M}} \bm{J} \bm{\mathcal{M}}^{-1}$ and $e^{\bm{\mathcal{L}}\tau}=\bm{\mathcal{M}} e^{\bm{J}\tau} \bm{\mathcal{M}}^{-1}$.

Our proposed learning procedure will evaluate the diffusion constants in the supra-Laplacian matrix $\bm{\mathcal{L}}$ as well as the $\bm{\Sigma}$ matrix. To that end, we proceed to minimize the Frobenius norm of the difference between $\textbf{X}$ and it's predict $\widehat{\textbf{X}}$, 
resulting from,
\begin{equation}
\label{eq:open5}
\begin{split}
\argmin_{\bm{\Sigma},\; D_1,..} g=||\bm{X}(t_1)-\widehat{\bm{X}}(t_1)||_F.\\
\end{split} 
\end{equation}
Solving this optimization problem helps us decompose the predicted matrix into two main components on the right-hand side of Eqn. (\ref{eq:open4}), the first term representing the interactions in the network, while the second quantifying the uncertainty which results from auxiliary inputs into the system.

\subsection{Diffusion Network Estimation (Learning the Supra-Laplacian Matrix)}
\label{subsec:EstimateL}
The supra-Laplacian matrix $\bm{\mathcal{L}}$ which we use in Eqn. (\ref{eq:open3}) for state prediction, is a result of the network connectivity (refer to Eqn. (\ref{eq:Sup_L})). In practice, hidden connections are pervasive, introducing uncertainty in the prediction, which are causing the information diffusion to require more than the predefined, explicit connections from the network. To that end, consider observations of $\textbf{X}(t)$ over $t\in[0,t_{1}]$, denote $\overline{\textbf{x}}(t):=vec(\textbf{X}(t))$, the vectorization of $\textbf{X}(t)$ to obtain a vector differential system in order to learn the supra-Laplacian matrix $\bm{\mathcal{L}}$ of Eqn. (\ref{eq:open3}):
\[
	\dot{\overline{\textbf{x}}}(t)=\bm{\Lambda}\overline{\textbf{x}}(t)+ \overline{\textbf{w}}(t)\;,\,\,0\leq t\leq t_{1},
\]
and we have $\bm{\Lambda}=\textbf{I}_{T}\otimes (-\bm{\mathcal{L}})$, the Kronecker product of $T$-by-$T$ identity matrix with $(-\bm{\mathcal{L}})$ and $\overline{\textbf{w}}(t)$ is the vectorization of $\textbf{w}(t)= \bm{\Sigma} \frac{d\textbf{B}(t)}{dt}$. 

We consider the simple cost function $J=\frac12\bm{\epsilon}^{T}\bm{\epsilon}$, where $\bm{\epsilon}=\overline{\textbf{x}}-\hat{\overline{\textbf{x}}}$, and hence for the estimation $\hat{\bm{\Lambda}}$ we have $\dot{\hat{\bm{\Lambda}}}=\gamma (\overline{\textbf{x}}-\hat{\overline{\textbf{x}}}) \overline{\textbf{x}}^{T}$ (derivative of $J$ with respect to $\bar{x}$), where the estimation $\hat{\overline{\textbf{x}}}(t+1)=e^{\hat{\bm{\Lambda}}} \overline{\textbf{x}}(t)$, and $\gamma>0$ is chosen appropriately as the scaling gain. In optimization iterations, the estimated value of $\hat{\bm{\Lambda}}$ at $i^{th}$ iteration is as follows:
\[
\hat{\bm{\Lambda}}_i=\hat{\bm{\Lambda}}_{i-1}+\dot{\hat{\bm{\Lambda}}}_{i-1}.
\]
We use $\hat{\bm{\Lambda}}_{0}=\textbf{I}_{T}\otimes (-\bm{\mathcal{L}})$, the graph Laplacian of the explicit following-follower network (as initialization). The learned $\bm{\Lambda}$ my however, not be exactly structured as $\textbf{I}_{T}\otimes (-\bm{\mathcal{L}})$, due to dependence of topics in the state space, as well as the nonlinearity and non-homogeneity of the diffusion. The resulting error $\bm{\epsilon}$ shall be considered for the estimation of the noise in the Kalman-Bucy filter as discussed next.

\subsection{A Refined Prediction: Kalman-Bucy Filtering}
\label{subsec:Kalman}
In prediction applications, the actual states of some of the nodes are sometimes known, and we want to predict those of all remaining nodes. An example of this may be seen in social networks, where state of the hub nodes, such as famous people or users with less restrictive privacy policies are known to the public, and one is interested in predicting the state of other less accessible users. 
Having partial knowledge of the states of a fraction of the nodes in the network, changes the state prediction problem to a Kalman predictor problem, and helps to refine the predicted states using a Kalman filter.
We propose a Kalman-Bucy filter as the optimal linear predictor for our system, and we write the observation equation as $\overline{\textbf{y}}(t)=(\textbf{I}_{T}\otimes \textbf{H})\overline{\textbf{x}}(t)+\overline{\textbf{v}}(t)$, with $\textbf{H}$ as a diagonal indicator matrix with 1 in all the observed entries, and 0 in all other entries.
$\bm{\Lambda}$ having been learned (see above Section), Kalamn-Bucy equations maybe written as:
\begin{gather}
\label{eq:kalman-1}
\notag
\dot{\overline{\textbf{x}}}(t)=\hat{\bm{\Lambda}}\overline{\textbf{x}}(t)+ \overline{\textbf{w}}(t), \\ \notag
\label{eq:kalman0}
\overline{\textbf{y}}(t)=\bm{\mathcal{H}} \overline{\textbf{x}}(t)+\overline{\textbf{v}}(t),
\end{gather}
where $\bm{\mathcal{H}} = \textbf{I}_{T}\otimes \textbf{H}$ and the noises $\overline{\textbf{w}}(t)$ and $\overline{\textbf{v}}(t)$ are zero-mean white (temporally) processes, i.e, $E(\overline{\textbf{w}}(t){\overline{\textbf{w}}(s)}^{T}) = \textbf{Q}_{t}\delta(t-s)$, $E(\overline{\textbf{v}}(t){\overline{\textbf{v}}(s)}^{T}) = \textbf{R}_{t}\delta(t-s)$, $E(\overline{\textbf{w}}(t){\overline{\textbf{v}}(s)}^{T})=0$. 
By considering small time intervals on discretization of the linear continues time system $(\delta_t=1)$, one can write the state equation as $\bar{\textbf{x}}(t+1)=\hat{\textbf{F}} \bar{\textbf{x}}(t)+ \bar{\textbf{w}}(t)$, where $\hat{\textbf{F}}= \textbf{I} + \hat{\bm{\Lambda}}$:
\begin{gather}
\label{eq:discret}
\notag
\overline{\textbf{x}}(t+\delta_t)\simeq\overline{\textbf{x}}(t)+ \delta_t \,\dot{\overline{\textbf{x}}}(t), \\ \notag
\overline{\textbf{x}}(t+\delta_t)\simeq\overline{\textbf{x}}(t)+ \delta_t \, (\hat{\bm{\Lambda}}\overline{\textbf{x}}(t)+ \overline{\textbf{w}}(t)), \\ \notag
\overline{\textbf{x}}(t+\delta_t)\simeq(\textbf{I}+ \delta_t \, \hat{\bm{\Lambda}})\overline{\textbf{x}}(t)+ \delta_t \, \overline{\textbf{w}}(t). \notag
\end{gather}
Discretizing the Kalman-Bucy equations, gives us following discrete time equations:
\begin{gather}
\label{eq:kalman1}
\bar{\textbf{x}}(t+1)=\hat{\textbf{F}} \bar{\textbf{x}}(t)+ \bar{\textbf{w}}(t), \\ 
\label{eq:kalman2}
\overline{\textbf{y}}(t)=\bm{\mathcal{H}} \overline{\textbf{x}}(t)+\overline{\textbf{v}}(t).
\end{gather}

Having Eqns. (\ref{eq:kalman1} and  \ref{eq:kalman2}) as the state and observation equations respectively, we can predict and refine the predicted states of the nodes using Algorithm (\ref{alg}):
\begin{algorithm}[]
\caption{}\label{alg}
\textbf{Learning phase:}
\begin{algorithmic}[1]
\State $\overline{\textbf{x}}(t) \gets vec(\textbf{X}(t))$
\State $\hat{\bm{\Lambda}} \gets \textbf{I}_{T}\otimes (-\bm{\mathcal{L}})$ \Comment{Initial state.}
\Repeat :
\State $\hat{\overline{\textbf{x}}}(t+1) \gets e^{\hat{\bm{\Lambda}}} \overline{\textbf{x}}(t)$
\State $\dot{\hat{\bm{\Lambda}}} \gets \gamma (\overline{\textbf{x}}-\hat{\overline{\textbf{x}}}) \overline{\textbf{x}}^{T}$
\State $\hat{\bm{\Lambda}} \gets \hat{\bm{\Lambda}} + \dot{\hat{\bm{\Lambda}}} $
\Until{$||\overline{\textbf{x}}-\hat{\overline{\textbf{x}}}||_2 < \eta.$} \Comment{Convergence criteria.}
\end{algorithmic}
\textbf{Kalman filter prediction on test data:}
\begin{algorithmic}[1]
\State $\textbf{R}_{e,t} \gets \textbf{R}_t + \bm{\mathcal{H}} \bm{\Pi}_{t|t-1}\bm{\mathcal{H}}^T$  \Comment{Updating.}
\State $\hat{\bar{\textbf{x}}}_{t|t} \gets \hat{\bar{\textbf{x}}}_{t|t-1} + \bm{\Pi}_{t|t-1}\bm{\mathcal{H}}^T \textbf{R}^{-1}_{e,t} [\bar{\textbf{y}}_t - \bm{\mathcal{H}} \hat{\bar{\textbf{x}}}_{t|t-1}]$
\State $\bm{\Pi}_{t|t} \gets \bm{\Pi}_{t|t-1} - \bm{\Pi}_{t|t-1}\bm{\mathcal{H}}^T \textbf{R}^{-1}_{e,t} \bm{\mathcal{H}} \bm{\Pi}_{t|t-1}$
\State $\hat{\textbf{F}} \gets \textbf{I} + \hat{\bm{\Lambda}}$
\State $\hat{\bar{\textbf{x}}}_{t+1|t} \gets \hat{\textbf{F}} \hat{\bar{\textbf{x}}}_{t|t}$ \Comment{Predicting.}
\State $\bm{\Pi}_{t+1|t} = \hat{\textbf{F}}\bm{\Pi}_{t|t}\hat{\textbf{F}}^T + \textbf{Q}_t$
\end{algorithmic}
\end{algorithm}

The "learning phase" of the Algorithm (\ref{alg}) is estimating the supra-Laplacian matrix $\hat{\bm{\Lambda}}$ (see above Section).
The second phase of the algorithm, is refining the estimated state of the nodes. Note that $\textbf{R}_{t}$ is the covariance of the observational error, and  $\hat{\bar{\textbf{x}}}_{t_2|t_1}$ denotes the linear prediction of $\bar{\textbf{x}}$ at time $t_2$ given observations up to and including time $t_1$.
The filter equation of a Kalman-Bucy filter\cite{kalman1960}, lines 1-3 of the algorithm, is given by:
\begin{gather}
\label{eq:kalman3}
\dot{\hat{\overline{\textbf{x}}}}=\hat{\bm{\Lambda}}\hat{\overline{\textbf{x}}}+\textbf{G}_{t}(\overline{\textbf{y}}(t)-\bm{\mathcal{H}}\hat{\overline{\textbf{x}}}(t)),\\
\label{eq:kalman4}
\textbf{G}_{t}=\bm{\Pi}_{t}\bm{\mathcal{H}}^T\textbf{R}_{t}^{-1},
\end{gather}
while the $\textbf{G}_{t}$ is the Kalman gain, and the state covariance $\bm{\Pi}_{t}$ satisfies the \emph{Riccati} equation:
\begin{equation}
\label{eq:kalman5}
\dot{\bm{\Pi}}_{t}=\hat{\bm{\Lambda}}\bm{\Pi}_{t}+\bm{\Pi}_{t}\hat{\bm{\Lambda}}^{T}+\textbf{Q}_{t}-\textbf{G}_{t}\textbf{R}_{t}\textbf{G}_{t}^{T}.
\end{equation}
For simplicity, we further assumed that the errors in the state prediction and observation are Gaussian processes. 

The designed algorithm shows the discrete time, state update of the Kalman predictor. The estimated states of the available nodes,  $\bm{\mathcal{H}}\hat{\overline{\textbf{x}}}(t)$, are compared with the state of the available nodes, $\overline{\textbf{y}}(t)$, as measurements observed over time, to evaluate the extent of statistical noise and other inaccuracies in predicting phase. The Kalman gain $\textbf{G}_{t}$, is tuned to assign accurate gain on the measurement (state of the available nodes) or follow the prediction model (estimating the state of the unknown nodes) more closely.
More simply, the algorithm recursively learns from the error appeared in the estimation of the states of the available nodes and-- by taking into account of the covariance of the measured error-- refines the estimation of the states of the less accessible nodes.

\subsection{Inter-layer Connectivity: Structural Robustness of an Interconnected Network}
\label{subsec:eigen} 
In an interconnected network setting, the inter-layer links play a crucial role in speed of diffusion between layers. Strong inter-layer links will cause a faster information diffusion among the layers and while, weak inter-layer connections will yield a set of independent layers \cite{sole2013spectral}. In this section we use perturbation theory \cite{_modern_2010} to study the effect of weak inter-layer linkage on the connectivity of the over-all interconnected network. The second smallest eigenvalue of a Laplacian matrix (algebraic connectivity) helps us to uncover how close is the interconnected network is to break into multiple connected components \cite{sole2013spectral, chung1997spectral, mohar1991laplacian}. To this end, we claim following proposition,

\begin{proposition}
If an interconnected network has connected intra-layer networks but weak inter-layer links, the second smallest eigenvalue of the supra-Laplacian matrix of that network is equal to,
\begin{equation}
\label{eq:prop2}
\epsilon \textbf{u}_n^T \bm{\mathcal{L}}_I \textbf{u}_n,
\end{equation}
where $\textbf{u}_n$ is an eigenvector of the intra-layer supra-Laplacian matrix $\bm{\mathcal{L}}_L$. 
\end{proposition}
\begin{proof}
Similarly to Eqn. (\ref{eq:L_Summation}), write the supra-Laplacian matrix as a supra-Laplacian matrix of the intra-layer connectivity $\bm{\mathcal{L}}_L$, and a supra-Laplacian matrix of a weak inter-layer $\epsilon \bm{\mathcal{L}}_I$:
\begin{equation}
\label{eq:Lperturbation}
\bm{\mathcal{L}} = \bm{\mathcal{L}}_L + \epsilon \bm{\mathcal{L}}_I,
\end{equation}
where $\epsilon$ is a small positive number. 

Denote $\textbf{v}$ and $\lambda$ by respectively an eigenvector and the associated eigenvalue of the supra-Laplacian matrix $\bm{\mathcal{L}}$, we have $\bm{\mathcal{L}} \textbf{v}= \lambda \textbf{v}$. Similarly for the intra-layer supra-Laplacian, i.e, $\bm{\mathcal{L}}_L \textbf{u}_n=\lambda_n \textbf{u}_n$, where $\{\textbf{u}_n\}$ is an orthonormal set of vectors such that: $\textbf{u}_n^T \textbf{u}_m=\delta_{nm}$. 

Since Laplacian matrices are symmetric and diagonalizable, an eigenvector of the supra-Laplacian matrix maybe written as a linear combination of the eigenvectors of the intra-layer supra-Laplacian ($\textbf{v}=\sum_m c_m \textbf{u}_m$). We therefore, have:
\begin{equation}
\label{eq:Lperturbation2}
\bm{\mathcal{L}} \sum_m c_m \textbf{u}_m = \lambda \sum_m c_m \textbf{u}_m.
\end{equation}
By right multiplying both hand-sides of the Eqn. (\ref{eq:Lperturbation2}) by $\textbf{u}_n^T$ and replacing the $\bm{\mathcal{L}}$ with Eqn. (\ref{eq:Lperturbation}) we will readily have: 
\begin{equation}
\label{eq:Lperturbation3}
c_n \lambda_n + \epsilon \sum_m c_m \textbf{u}_n^T \bm{\mathcal{L}}_I \textbf{u}_m = \lambda c_n.
\end{equation}
Since we know that $\textbf{v}= \textbf{u}_n + \textbf{o}(\epsilon)$, the eigenvector $\textbf{v}$ should be mainly in the same direction of one of the eigenvector's of the $\bm{\mathcal{L}}_L$ matrix with small perturbation. Therefore we can assume $c_m=o(\epsilon)$ and $c_n=o(1)$. In Eqn. (\ref{eq:Lperturbation3}), if we ignore the terms which have both $\epsilon$ and $c_m$ multipliers we will get: 
\begin{equation}
\label{eq:Lperturbation4}
\lambda = \lambda_n + \epsilon \textbf{u}_n^T \bm{\mathcal{L}}_I \textbf{u}_n.
\end{equation}

Since $\bm{\mathcal{L}}_L$ is the Laplacian matrix of $M$ independent intra-layers, it will have at least $M$ unconnected components and $M$ zero eigenvalues \cite{sole2013spectral}. From Eqn. (\ref{eq:Lperturbation4}), we therefore can infer that 
if an interconnected network has connected intra-layer networks but weak inter-layer links, the second smallest eigenvalue of the supra-Laplacian matrix of that network is equal to $\epsilon \textbf{u}_n^T \bm{\mathcal{L}}_I \textbf{u}_n$.

\end{proof}

 In section \ref{sec:Results} we demonstrate through an experiment the effect of weak inter-layer connectivity of our state predictions model.
\section{Experiments}
\label{sec:Results}
To evaluate and substantiate the theoretical interconnected-network model proposed above section, we conduct discrete-time simulations of our system in both "closed and open" system conditions. 
\subsection{Data Sets}
\label{subsec:Datasets}
Our experiments have been carried out using  the following data sets,
\begin{itemize}
 \item \textbf{Network of professors and publications $(N \sim 100)$:}
 We assume a three-layer network, with two agent-layers for the professors (the first and the second layers) and one information-layer for publications (the third layer).
In this data-set the agents are $79$ professors at North Carolina State University, and the documents are $1000$ abstracts of academic papers published by these professors 1990 to 2014.
The agents in the first and the second agent-layers are the same individuals with one-to-one connections between the agents. 
The first agent-layer is based on the number of  papers two professors have co-authored in the same venue. $\textbf{W}_{i,j}^{(1)}$ is the cumulative sum of  papers published by professor $i$ and $j$ in the same venue.
The second agent-layer reflects research group co-membership of two professors (with 8 different research groups considered).
\begin{equation}
\label{eq:res2}
\textbf{W}_{i,j}^{(2)}=  \begin{cases}
   1& \text{If professor $i$ and $j$ are in the same group,}\\
   0& \text{Otherwise.}
  \end{cases}
\end{equation}
The third layer (the information-layer) is based on the topical similarity of the produced documents, and is quantified by the inverse distance between the topical vectors $(T=10)$ of the documents. An $\epsilon$-neighborhood criterion has been applied to break the weaker links,
\begin{equation}
\label{eq:res3}
\textbf{W}_{i,j}^{(3)}=  \begin{cases}
   \frac{1}{||\textbf{x}^{(3)}_i-\textbf{x}^{(3)}_j||_2}& \text{If larger than $\epsilon$,}\\
   0& \text{Otherwise.}
  \end{cases}
\end{equation}
$T=10$ dimensional topic vectors of documents have been produced by performing the LDA  topic modeling algorithm \cite{Blei:2003:LDA} on the raw text documents.

 \item \textbf{Network of Twitter users and Hashtags $(N = 1000)$:}
A two-layer network of Twitter users, (first layer) with $(N = 1000)$ users  and the similarity network between eight different popular Hashtags used by the Twitter users in June 2009. The Hashtags are as follows: \#jobs, \#spymaster, \#neda, \#140mafia, \#tcot, \#musicmonday, \#Iranelection, \#iremember.
The agent-layer network is a directed graph, based on who is following whom on Twitter.
The information-layer network is the similarity network between the Hashtags, where  two Hashtags are considered similar if both have been mentioned in the same Tweet. We have used the Jaccard index \cite{Jaccard} to build the network between the Hashtags, and the $\epsilon$-neighborhood criterion has been applied to break the weaker links.
\begin{equation}
\label{eq:res3}
\textbf{W}_{m,n}^{(2)}=
\begin{cases}
   \frac{|\text{Tweets which have both m and n}|}{|\text{Tweets which have m or n}|} & \text{If greater than }\epsilon.\\
   0& \text{Otherwise.}
  \end{cases}
\end{equation}

 \item \textbf{Network of Twitter users and Hashtags $(N = 5000)$:}
An exactly similar setting as the previously mentioned data set, with $N = 5000$ users.
 
\end{itemize}
\subsection{Analysis of Results}
\label{subsec:AnalysisofResults}
Fig. (\ref{fig:res_NCreach}) shows  the results of our experiment on this data set (the network of professors and publications). The x-axis represents time, (and more specifically  the consecutive years from  Year 2000 until  Year 2014). In our experiment, our goal is to  predict the topical state of all the agents in each year by having the topical state of the agents in the previous year. More formally, we seek  $\widehat{\textbf{X}}^{(1)}(t)$ given $\widehat{\textbf{X}}^{(1)}(t-1)$. 
\begin{equation}
\label{eq:ER_measure}
Error\;measure(t)=\frac{||\widehat{\textbf{X}}^{(1)}(t)-\textbf{X}^{(1)}(t)||_F}{||\textbf{X}^{(1)}(t)||_F},
\end{equation}
where $\textbf{X}^{(1)}(t)$ is the ground truth matrix.
The data over years 1990-1999 are used to learn the diffusion constants, which those from years 2000-2014 for testing. 
The graph in blue is the prediction errors by just considering a single-layer co-authorship network of agents/professors. The graph in red reflects the prediction error by considering a three-layer interconnected network of the heterogeneous nodes (as described in Section \ref{subsec:Datasets}). The graph in black shows the resulting changes  in the topical states of the agents. 
\begin{figure}
\centering
\includegraphics[width=0.33\textwidth]{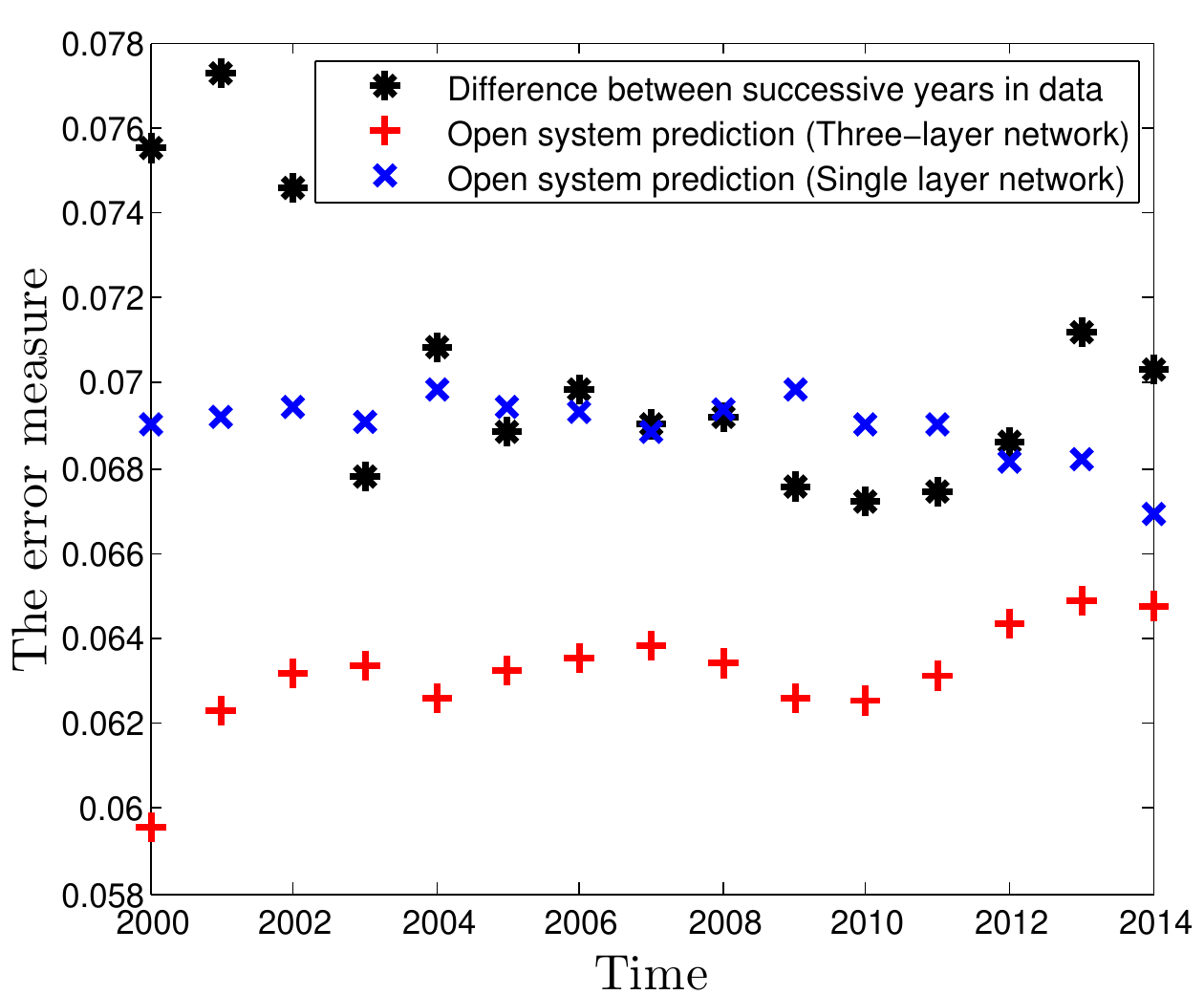}
\caption{Experiment for College Professor Network with 1000 publication documents.}
\label{fig:res_NCreach}
\end{figure}
This is hence a measure to reflect the difference in the yearly topical states of the agents, as in,
\begin{equation}
\label{eq:diff}
\frac{||\textbf{X}^{(1)}(t)-\textbf{X}^{(1)}(t-1)||_F}{||\textbf{X}^{(1)}(t)||_F}.
\end{equation}
Eqn. (\ref{eq:diff}) is an upper bound of our prediction error, and any reasonable prediction method should have a smaller error than the the one in Eqn. (\ref{eq:diff}). We may indeed just assume  no change in the agent states has occured to achieve the error in Eqn. (\ref{eq:diff}).

\subsubsection{Prediction using single layer network vs. heterogeneous network}
\label{subsec:result_singleLayervsMultilayer}
As may be seen in Fig. (\ref{fig:res_NCreach}), the prediction method based on a three-layer network achieves a lower error than the prediction based on a single-layer network. Note, the single-layer network does not help in predicting the topical states of the agents. The reason is that there are only 79 agents in this experiment and the co-authorship network between the agents is not particularly suited to predict the future state of the agents.
 The three-layer network, on the other hand, considers the similarity between the documents present in the network, and accounts for  more elaborate diffusion paths, thus enhancing the prediction phase.
 The average prediction improvement by the three-layer network over that of a   single-layer network is 8 percent.

\begin{figure}
\centering
\includegraphics[width=0.33\textwidth]{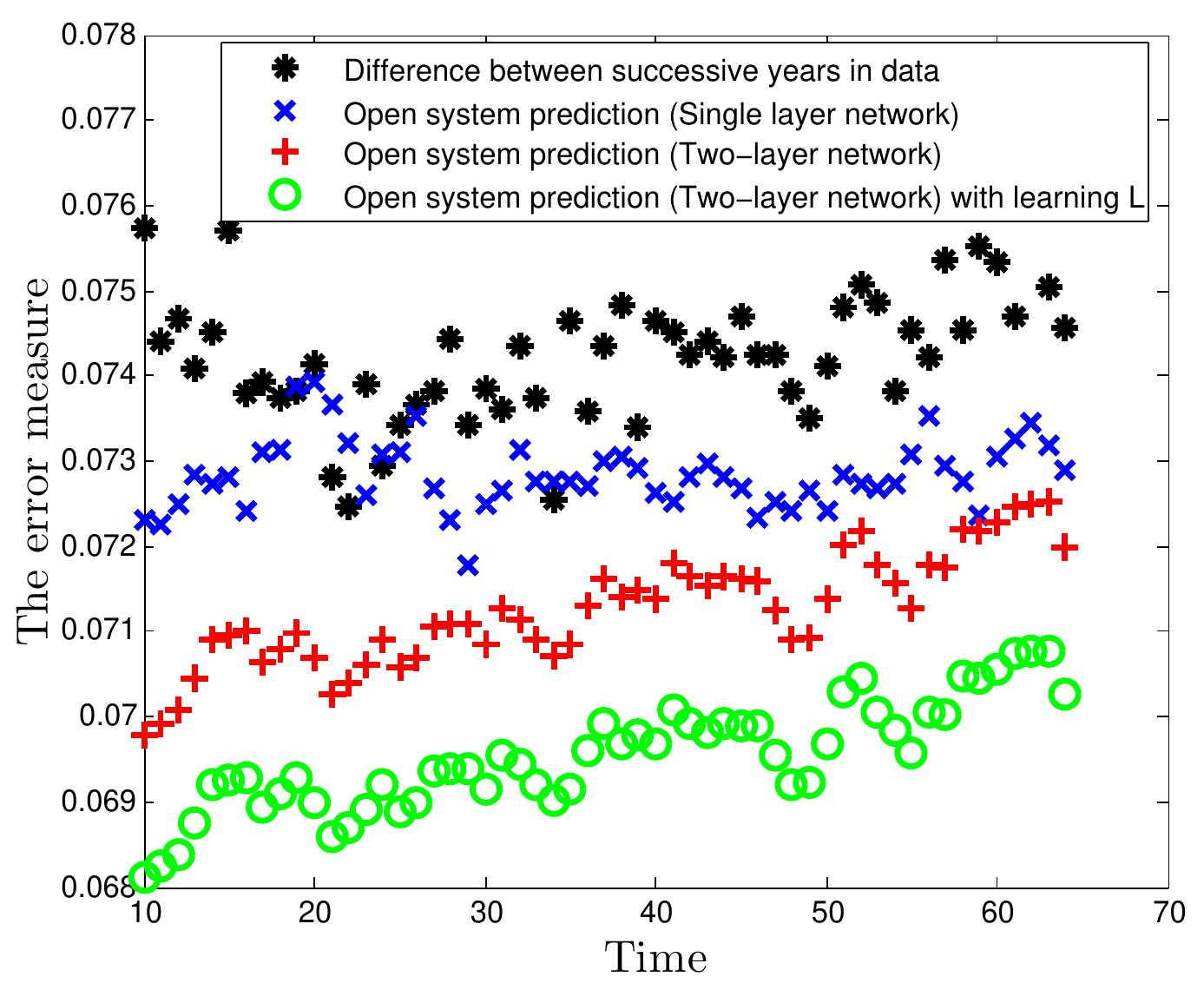}
\caption{Experiment for a  Twitter network with 1000 agents and 8 Hashtags.}
\label{fig:res_tweet1000}
\end{figure}
Figure \ref{fig:res_tweet1000} displays the results of  our experiment on the second data set. The time (x-axis) reflects intervals of consecutive six-hour periods. In this experiment, recall, our goal is  to predict the topical state of all agents at each time point when given  the topical state of the agents at previous time points. The error  metric is similar to that in  Eqn. (\ref{eq:ER_measure}). The graph in blue reflects the prediction error by just considering a single-layer Twitter follower/following network between users. The graph in red displays the prediction error using a two-layer interconnected network of heterogeneous nodes (as  described in Section \ref{subsec:Datasets}). The graph  in black represents the changes  in the topical states of users (Eqn. (\ref{eq:diff})). Finally, the graph in green shows the prediction error by first estimating the Laplacian matrix using the learning dataset and subsequently using it for the prediction phase.
As may be seen in the figure, the prediction based on the two-layer network displays improvement over the prediction based on the single-layer network. 
We have used the data from  time point 1 until  time point 9 for learning the diffusion constants, and have used the data from  Time-Point 10 to 64 for testing. 
The average prediction improvement using a two-layer network over a single-layer is 6 percent, while the prediction improvement  using a single layer network over successive differences of data  is  3 percent. 
We expect the improvement to further increase as the amount of data increases.
Predicting the Laplacian matrix helps us on decreasing the prediction error. The average prediction improvement by this model over the single-layer method is about 10 percent.
\begin{figure}
\centering
\includegraphics[width=0.35\textwidth]{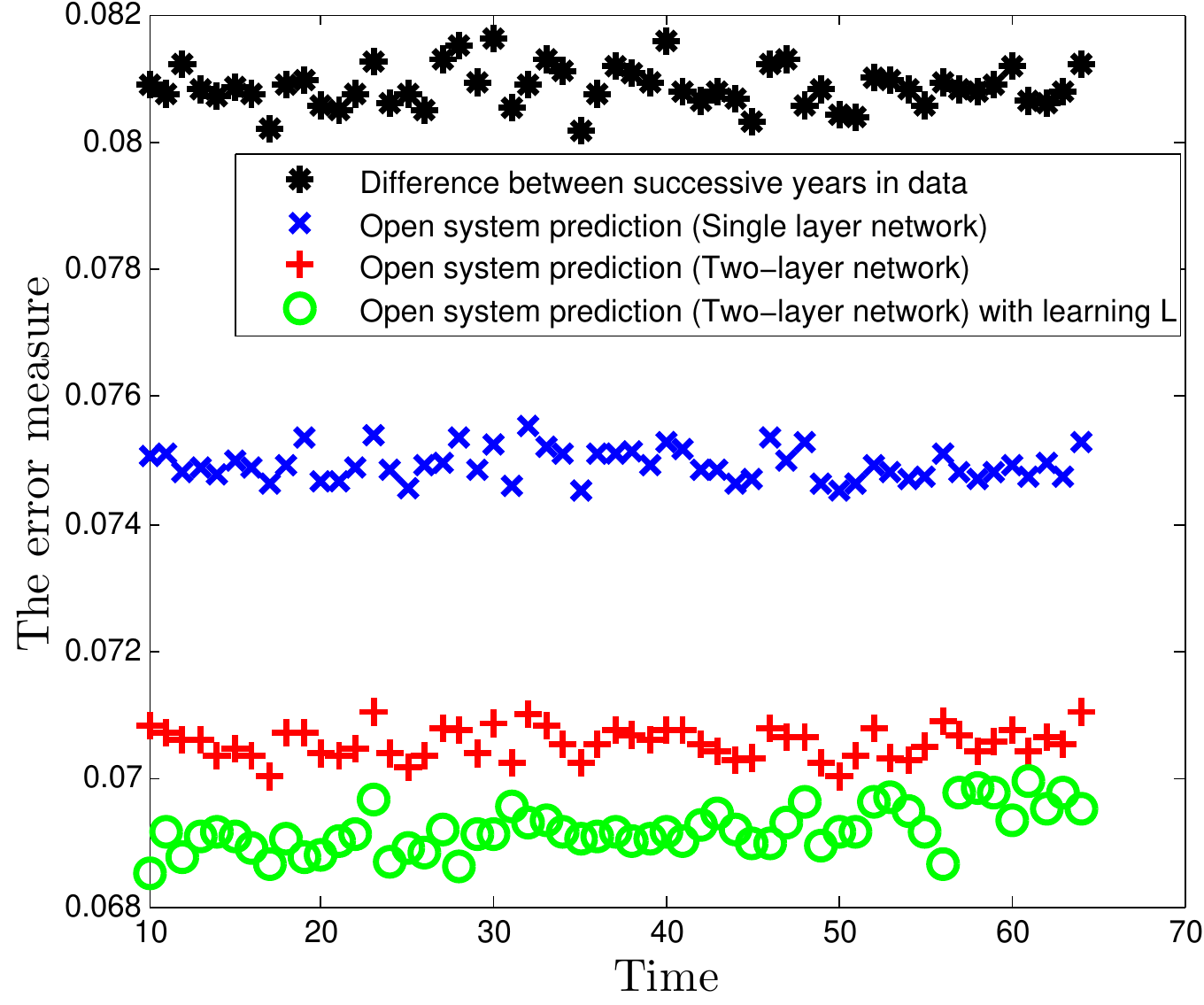}
\caption{Experiment over Twitter network with 5000 agents and 8 Hashtags.}
\label{fig:res_tweet5000}
\end{figure}
As in  Fig. (\ref{fig:res_NCreach}),  the better prediction performance of   the two-layer network is due to the topical similarities of the documents. This may be attributed to the fact that this additional layer helps us  discover certain topical interest of the users  using  similar Hashtags. For example, a Twitter user who has used a Hashtag \#Iranelection, is probably going to use a Hashtag \#neda (someone who died during the protests, after Iran's presidential election in 2009) later.
The lower prediction error by estimating the Laplacian matrix shows that the actual information diffusion structure can be different from the fixed connectivity imposed from the network.

Fig. (\ref{fig:res_tweet5000}) is the result of an experiment with  similar conditions to that of Fig. (\ref{fig:res_tweet1000}), if not for the number of agents being 5000. As may be  seen in the figure, the increased number of  agents improves the prediction error of the two-layer network relative to the error upper bound. This effect is mainly due to  two factors, (i) The increase in the number of  agents in the network improve the information diffusion and the mixing. (ii) The increased network size provides a better estimate of the diffusion constants. The average prediction improvement achieved by  the two-layer network is about 13 percent. The prediction by first estimating the Laplacian matrix, has about 15 percent improvement over the single layer prediction method. 

\begin{figure}[b]
\centering
\includegraphics[width=0.4\textwidth]{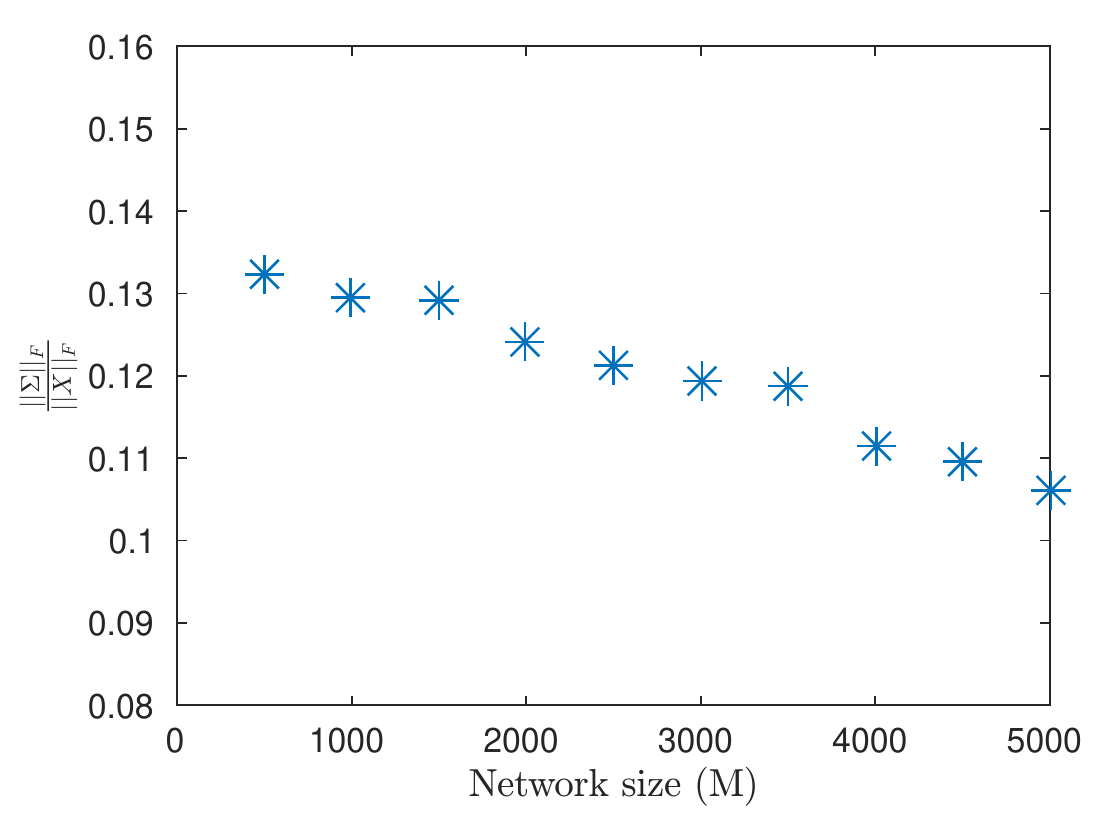}
\caption{Experiment over the Twitter network with different sizes.}
\label{fig:res_external}
\end{figure}
In Fig. (\ref{fig:res_external}), we have used the term $\frac{||\bm{\Sigma}||_F}{||\textbf{X}_0||_F}$ (with $\textbf{X}_0=\textbf{X}(0)$) to show the role of network size (number of nodes), on the extent to which the external sources affect the nodes. As the value of $\frac{||\bm{\Sigma}||_F}{||\textbf{X}_0||_F}$ term increases the second term on RHS of Eqn. (\ref{eq:open2}) will have a higher weight in determining the state of the nodes and consequently, the states of the agents will change by higher variance. 
The figure suggests that as the network size increases the external influence of input decreases. This result can be attributed to expanding the system borders into external areas. In other word, as we increase the network size, we are including a larger number of  factors which in turn, decreases the impact of auxiliary input.

\subsubsection{Kalman Prediction Filtering}
\label{subsec:result_KalmanFilter}
\begin{figure}[!t]
\centering
	\begin{subfigure}[t]{0.24\textwidth}
		\centering
		\includegraphics[width=0.9\textwidth]{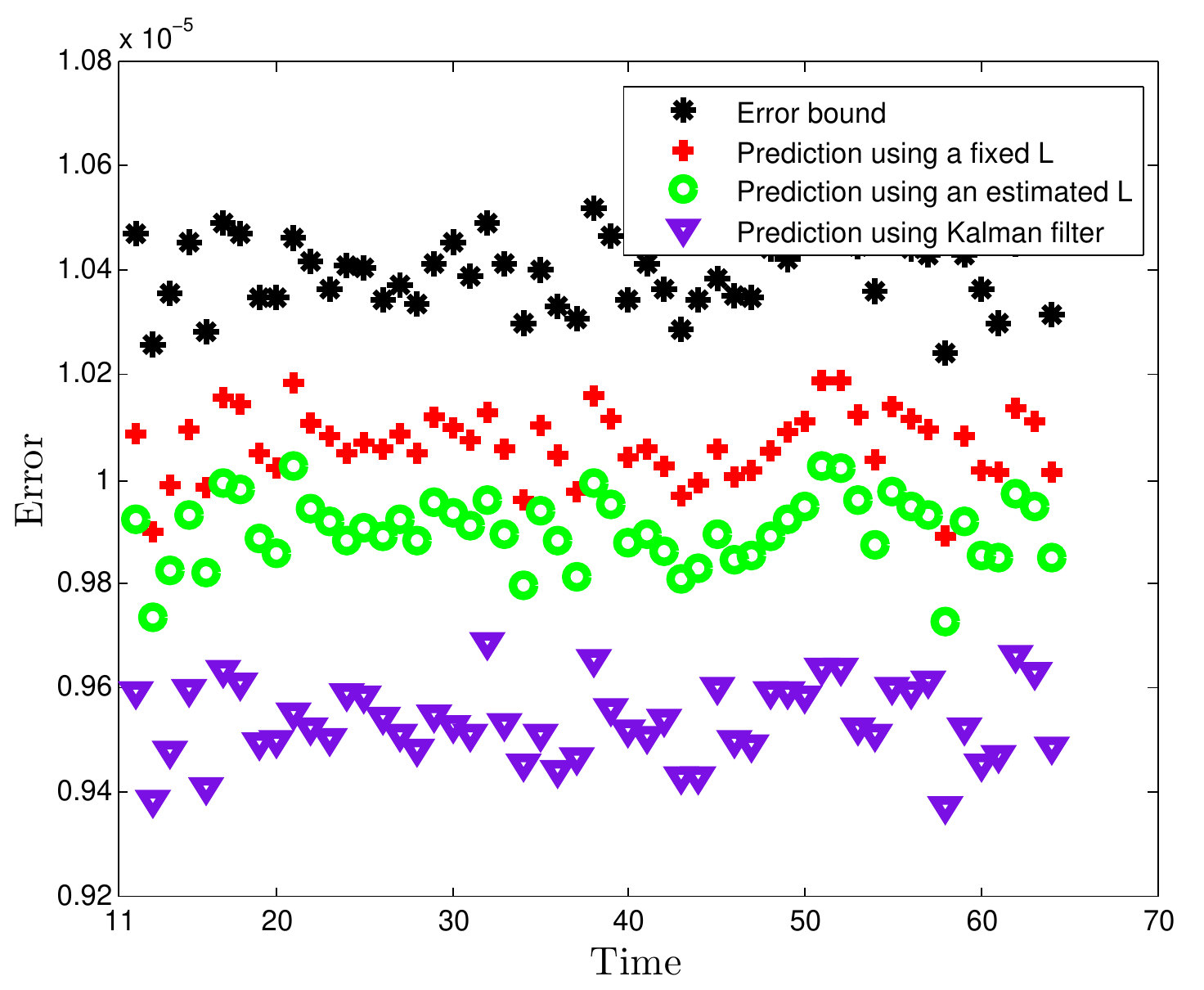}
		\caption{}
		\label{fig:kalman10}
	\end{subfigure} 
	\begin{subfigure}[t]{0.24\textwidth}
		\centering
		\includegraphics[width=0.9\textwidth]{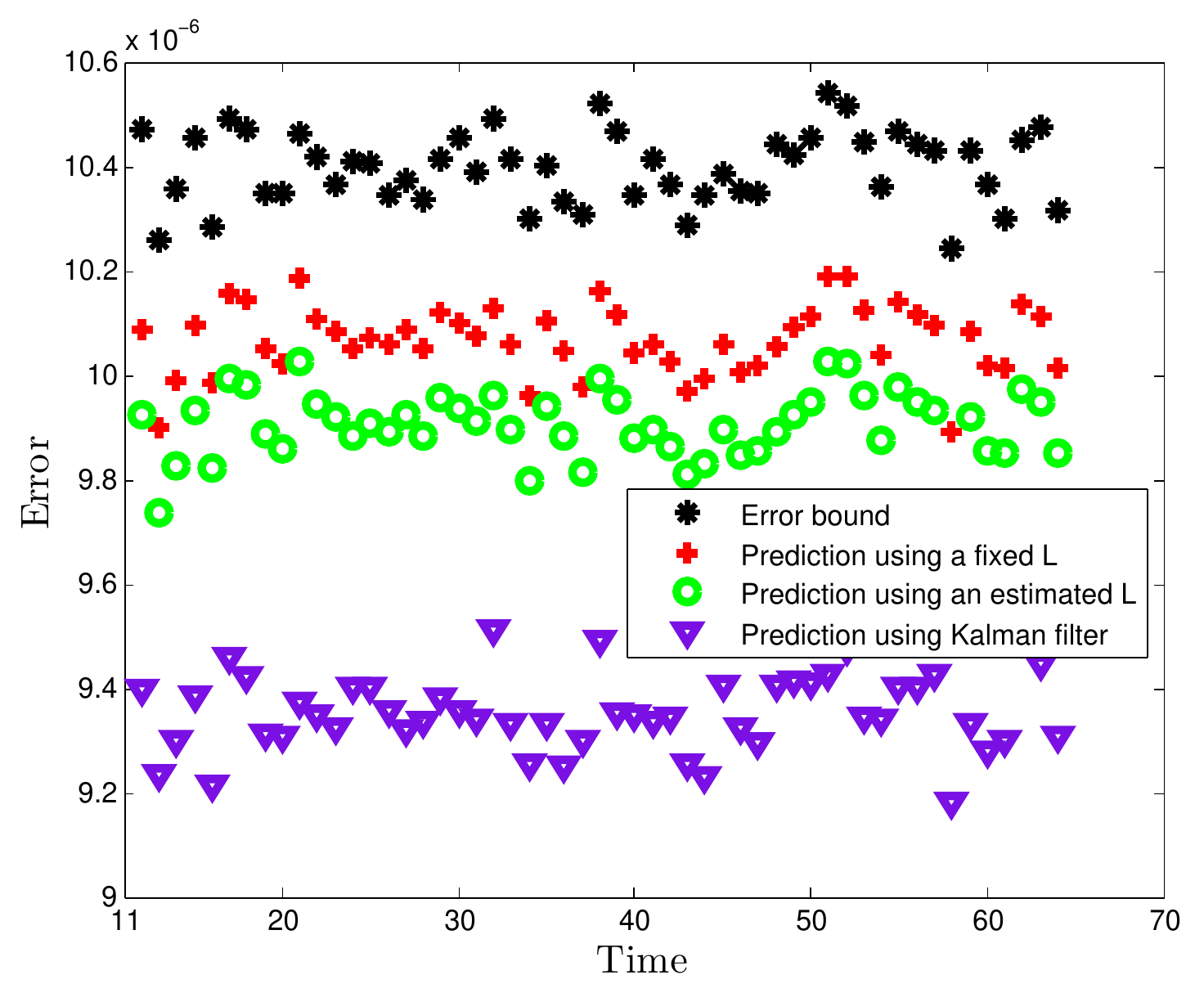}
		\caption{}
		\label{fig:kalman15}
	\end{subfigure}\\
	\begin{subfigure}[t]{0.24\textwidth}
		\centering
		\includegraphics[width=0.9\textwidth]{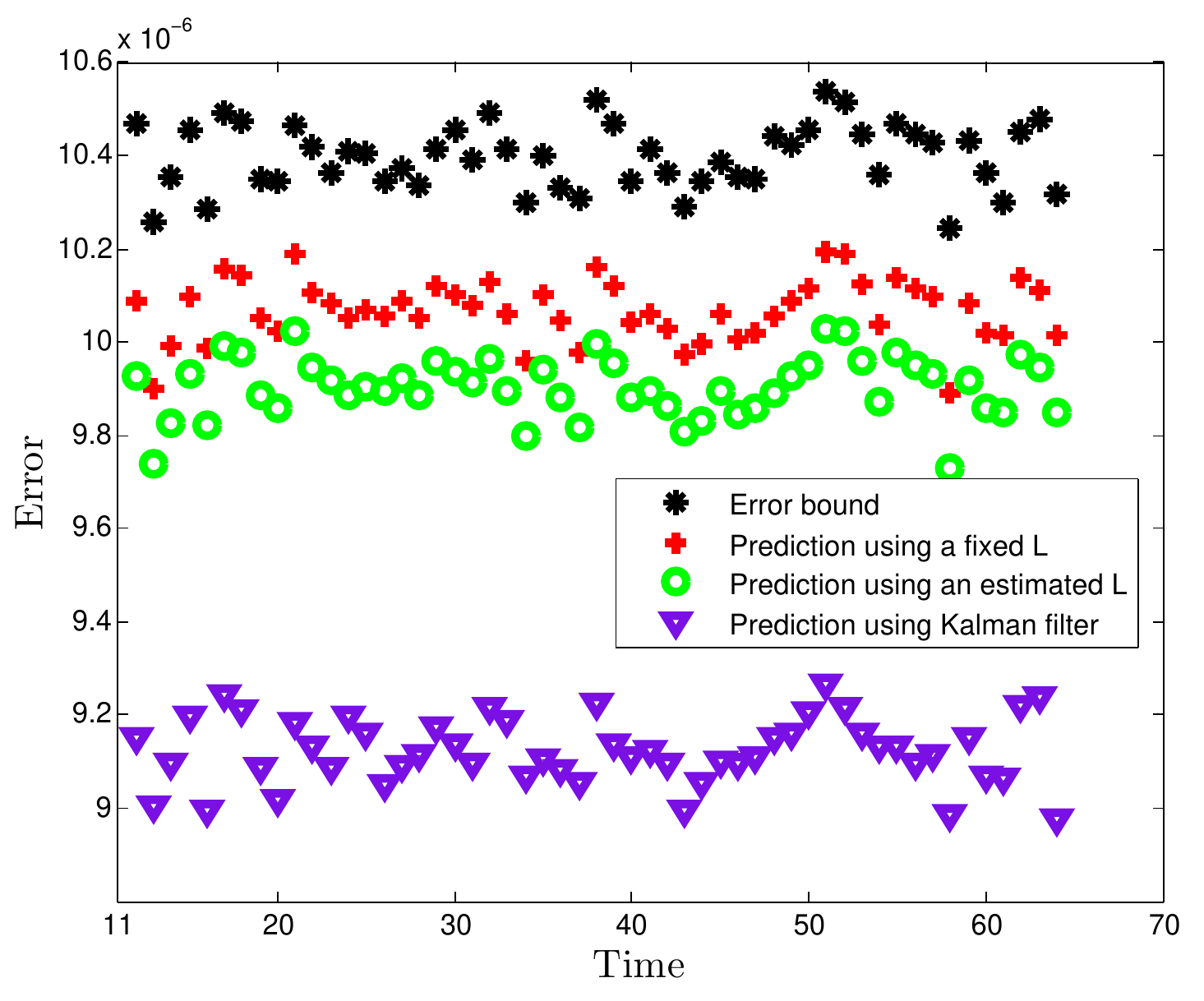}
		\caption{}
		\label{fig:kalman20}
	\end{subfigure} 
	\begin{subfigure}[t]{0.24\textwidth}
		\centering
		\includegraphics[width=0.9\textwidth]{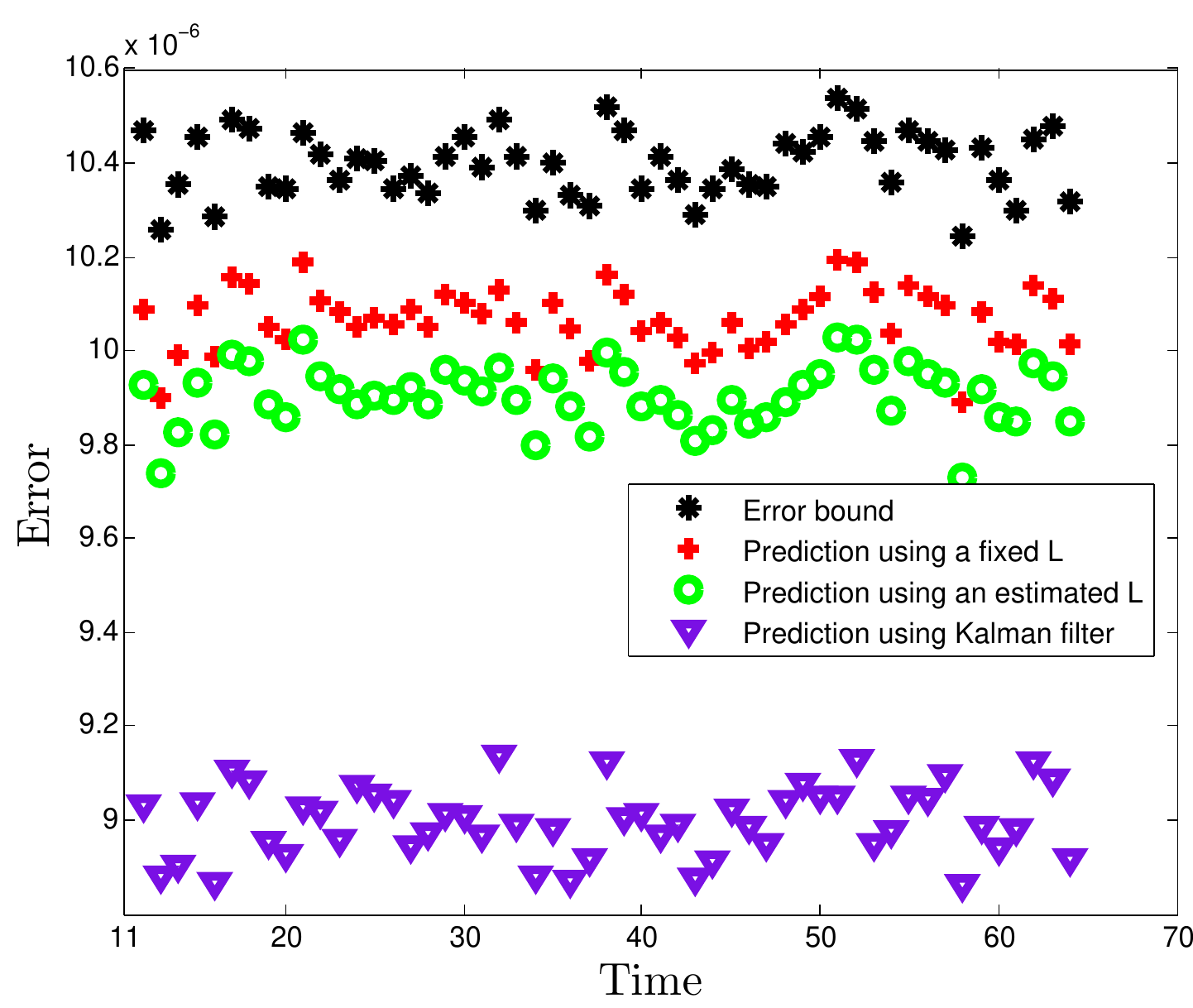}
		\caption{}
		\label{fig:kalman25}
	\end{subfigure}
	\caption{Experiment over Twitter network with 300 agents. Predicting the state of the agents using a fixed Laplacian matrix, using an estimated Laplacian matrix, and Prediction using Kalman filter with 10 percent, 15 percent, 20 percent, and 25 percent of observation of state of all the agents in Figs. (a), (b), (c) and (d) respectively.}
	\label{fig:kalmanResutl}
\end{figure}
Fig. (\ref{fig:kalmanResutl}) displays the results of an experiment in a small single layer dataset with 300 Twitter agents available. In this experiment, our goal is  to predict the topical state of all agents at each time point, when given the topical state of agents at previous time points. 

The black graphs in Fig. (\ref{fig:kalmanResutl}) show the resulting variation in the topical states of the agents overtime. 
This graph is an upper bound of our prediction error, and any reasonable prediction method should have a smaller error than this graph. We 
The graphs in red reflect the prediction error by just considering a fixed Laplacian matrix constructed from a Twitter follower/following network between agents. 
The graphs in green display the prediction error by first learning the Laplacian matrix from the learning data set, and then using the estimated Laplacian matrix in the prediction. We have used the data from time-point 1 until time-point 10 for learning the diffusion structure (the Laplacian matrix), and have used the data from time-point 11 to 64 for testing. Finally, the graphs in purple represent the prediction error with a Kalman predictor, knowing a fraction of the states of agents a priori.

Figs. (\ref{fig:kalman10}), (\ref{fig:kalman15}), (\ref{fig:kalman20}) and (\ref{fig:kalman25}) are the same experiments with different observation sizes of 10 percent, 15 percent, 20 percent and 25 percent of state of all the agents in the network respectively. 
As may be seen in the figures, the prediction error based on the estimated Laplacian matrix yields a lower error than fully trusting the connectivity structure in the network. 
This as expected, is due to static connectivity network (usually demonstrates the physical or online relation among the agents), falling short on effecting the actual underlying diffusion structure on the network.
Having a prior partial knowledge of the agents enabled us to use a Kalman predictor to further refine our prediction. As we increase the size of the observation set, in Figs. (\ref{fig:kalman15}), (\ref{fig:kalman20}) and (\ref{fig:kalman25}) we get a smaller error in the prediction. This effect is due to the Kalman gain which refines the state of the system by the estimation error of the agent's states.

\subsubsection{Algebraic connectivity of an interconnected network}
\label{subsec:result_Algebraic}
\begin{figure}[t]
\centering
\includegraphics[width=0.35\textwidth]{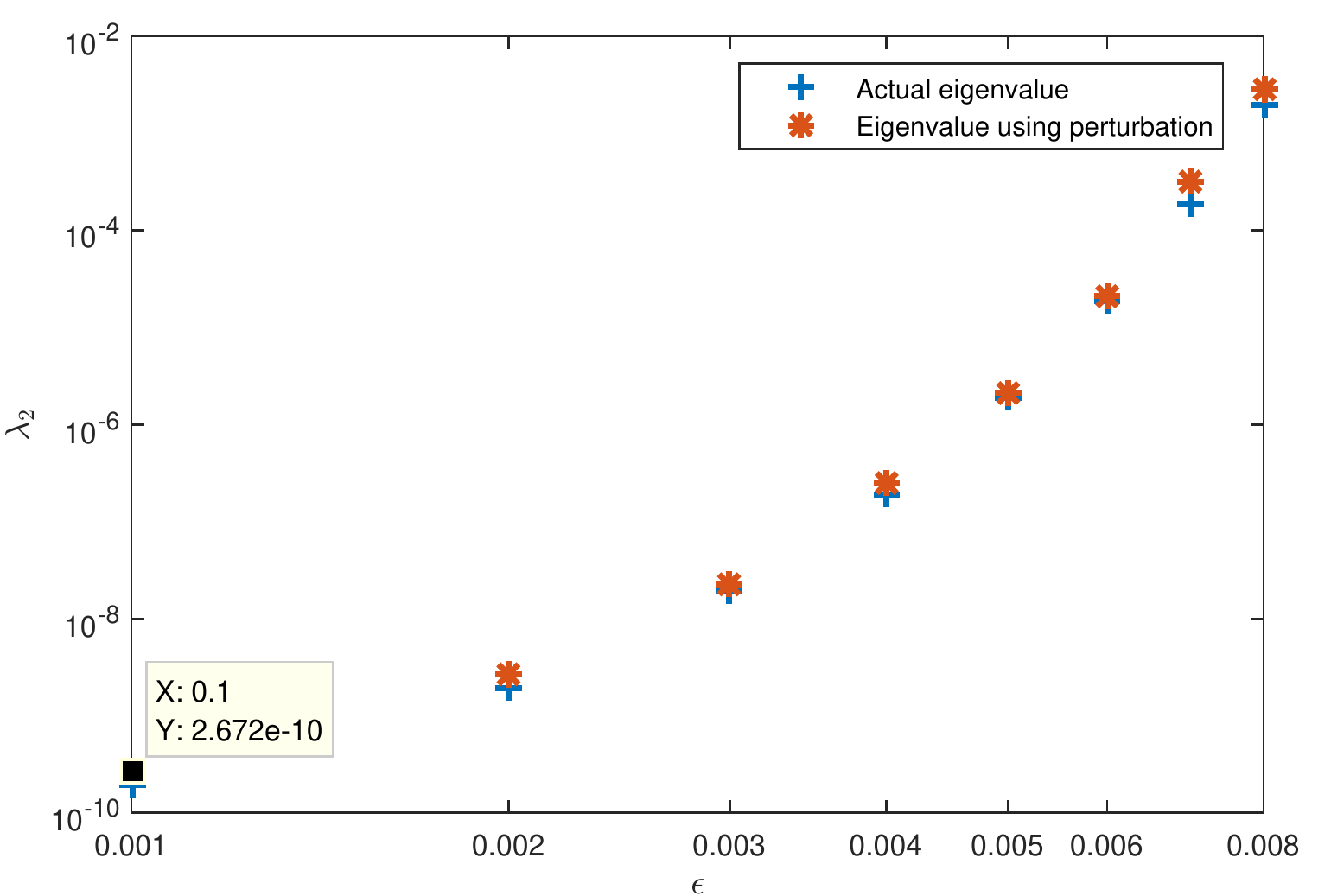}
\caption{Comparing the estimated value of the second smallest eigenvalue with the actual second smallest eigenvalue of the interconnected network.}
\label{fig:V2_compare}
\end{figure}
Fig. (\ref{fig:V2_compare}) compares the actual second smallest eigenvalue (so called fiddler's constant) of Laplacian matrix of an interconneted network with the estimated value using Eqn. (\ref{eq:Lperturbation4}). The x-axis shows different values of the $\epsilon$, and the y-axis shows the second smallest eigenvalues of the supra-Laplacian matrix associated.
As may be seen from the figure the estimated values closely match the actual values.
Fig. (\ref{fig:V2_compare2}) shows an experiment on the second dataset with different values of $\epsilon$, and the effect of weak inter-layer connectivity on the state prediction error. As seen in Fig. (\ref{fig:V2_compare2}), for small values of $\epsilon$, the interconnected network is a single layer network and the prediction error is close to the single layer prediction error. However, as we increase the value of the $\epsilon$, the inter-layer connectivity is getting stronger, and the prediction error decreases to the prediction error of a multi-layer network.

Based on Figs. (\ref{fig:V2_compare}, \ref{fig:V2_compare2}), one can use the Eqn. (\ref{eq:Lperturbation4}) to understand the strength of inter-layer connectivity in an interconnected network, and compute the algebraic connectivity. The algebraic connectivity in this case, helps us understand if the inter-layer connectivity is sequentially connecting all the intra-layers.

\begin{figure}[t]
\centering
\includegraphics[width=0.35\textwidth]{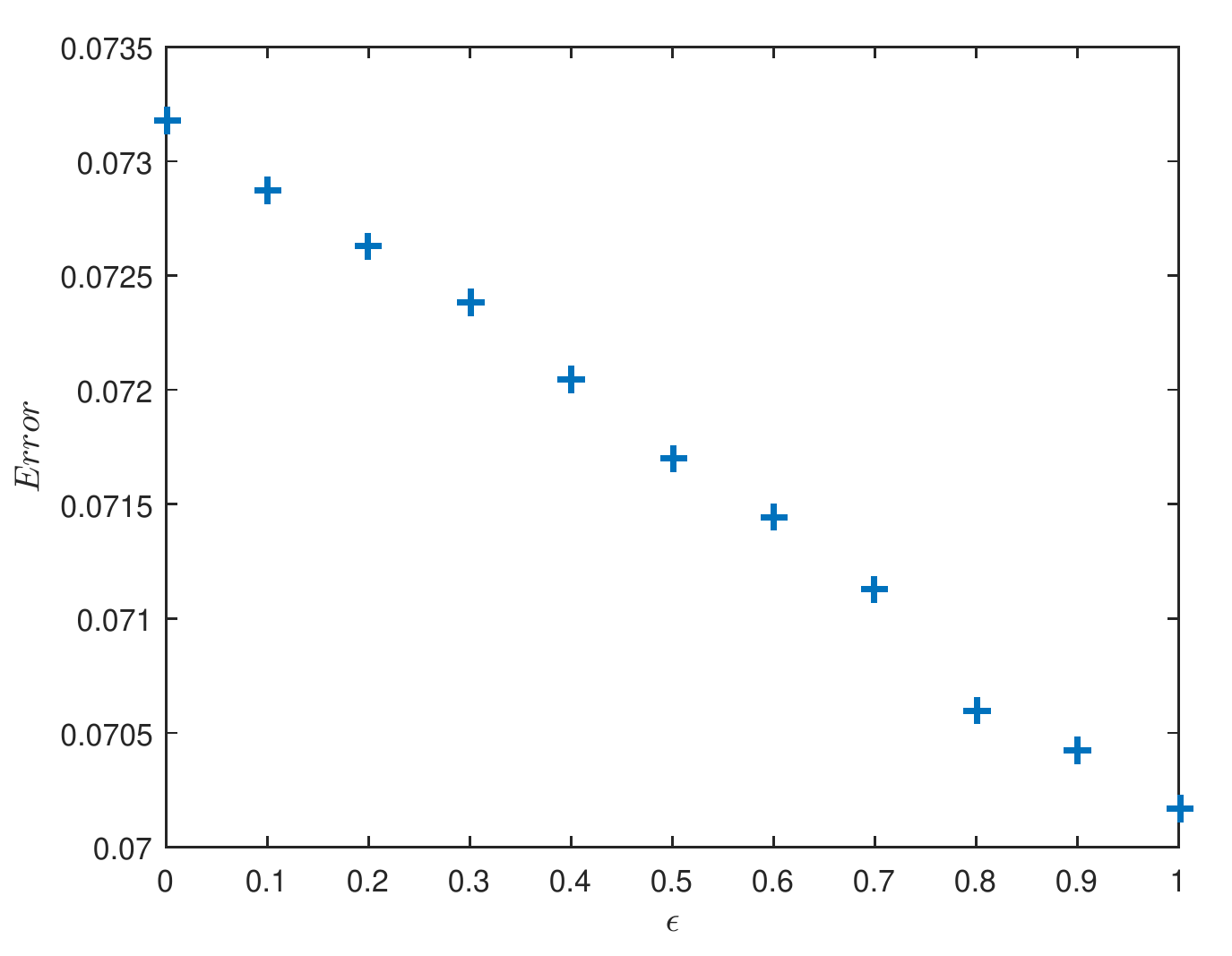}
\caption{The effect of weak inter-layer connectivity on the state prediction error.}
\label{fig:V2_compare2}
\end{figure}
\section{Conclusion}
\label{sec:conclusion}
In this work, we address the dynamics of an  inter-connected network by modeling changes in the topical states of agents using diffusion processes in the network. The intuition behind this model lies in the information diffusing on account of the possible interactions of the agents, which in turn may change their topical states over time. To further generalize this idea, we propose to follow the different interactions of the agents in separate layers. We also consider the similarity of the data flowing in the network as another network layer that may help explain information diffusion in the interconnected network. To this end, we proposed a diffusion equation for the three-layer interconnected network. Moreover, we consider the external effect on each node by assuming that the whole system is a massive Brownian particle. In this model, the changes in the state of each node is a function of its interactions with other nodes and with the external effect modeled as Brownian motions. We test our method by experimenting on three real-world data sets. The results show that, the prediction for an interconnected network achieves a lower error. This is due to explicitly accounting for the data of the network as it also evolves among the agents. Also, we also showed that increasing the size of the network yields an improvement in the error.



\bibliography{bib}

\begin{bibdiv}
\begin{biblist}

\bib{Blei:2003:LDA}{article}{
      author={Blei, David~M.},
      author={Ng, Andrew~Y.},
      author={Jordan, Michael~I.},
       title={Latent dirichlet allocation},
        date={2003-03},
        ISSN={1532-4435},
     journal={J. Mach. Learn. Res.},
      volume={3},
       pages={993\ndash 1022},
         url={http://dl.acm.org/citation.cfm?id=944919.944937},
}

\bib{active_passive_2}{article}{
      author={Brummitt, Charles~D.},
      author={Lee, Kyu-Min},
      author={Goh, K.-I.},
       title={Multiplexity-facilitated cascades in networks},
        date={2012Apr},
     journal={Phys. Rev. E},
      volume={85},
       pages={045102},
}

\bib{chung1997spectral}{book}{
      author={Chung, Fan~RK},
       title={Spectral graph theory},
   publisher={American Mathematical Soc.},
        date={1997},
      volume={92},
}

\bib{Faryad2014PhysRevE}{article}{
      author={Darabi~Sahneh, Faryad},
      author={Scoglio, Caterina},
       title={Competitive epidemic spreading over arbitrary multilayer
  networks},
        date={2014Jun},
     journal={Phys. Rev. E},
      volume={89},
       pages={062817},
}

\bib{Deer:1990:LSI}{article}{
      author={Deerwester, Scott~C.},
      author={Dumais, Susan~T},
      author={Landauer, Thomas~K.},
      author={Furnas, George~W.},
      author={Harshman, Richard~A.},
       title={Indexing by latent semantic analysis},
        date={1990},
     journal={JASIS},
      volume={41},
      number={6},
       pages={391\ndash 407},
}

\bib{SIR_M2}{article}{
      author={Dickison, Mark},
      author={Havlin, S.},
      author={Stanley, H.~E.},
       title={Epidemics on interconnected networks},
        date={2012Jun},
     journal={Phys. Rev. E},
      volume={85},
       pages={066109},
}

\bib{escolano2012heat}{article}{
      author={Escolano, Francisco},
      author={Hancock, Edwin~R},
      author={Lozano, Miguel~A},
       title={Heat diffusion: Thermodynamic depth complexity of networks},
        date={2012},
     journal={Physical Review E},
      volume={85},
      number={3},
       pages={036206},
}

\bib{Funk2010PhysRevE}{article}{
      author={Funk, Sebastian},
      author={Jansen, Vincent A.~A.},
       title={Interacting epidemics on overlay networks},
        date={2010Mar},
     journal={Phys. Rev. E},
      volume={81},
       pages={036118},
}

\bib{gomez2013diffusion}{article}{
      author={Gomez, Sergio},
      author={Diaz-Guilera, Albert},
      author={Gomez-Garde{\~n}es, Jesus},
      author={Perez-Vicente, Conrad~J},
      author={Moreno, Yamir},
      author={Arenas, Alex},
       title={Diffusion dynamics on multiplex networks},
        date={2013},
     journal={Physical review letters},
      volume={110},
      number={2},
       pages={028701},
}

\bib{gomez2010inferring}{inproceedings}{
      author={Gomez~Rodriguez, Manuel},
      author={Leskovec, Jure},
      author={Krause, Andreas},
       title={Inferring networks of diffusion and influence},
organization={ACM},
        date={2010},
   booktitle={Proceedings of the 16th acm sigkdd international conference on
  knowledge discovery and data mining},
       pages={1019\ndash 1028},
}

\bib{gomez2013structure}{inproceedings}{
      author={Gomez~Rodriguez, Manuel},
      author={Leskovec, Jure},
      author={Sch{\"o}lkopf, Bernhard},
       title={Structure and dynamics of information pathways in online media},
organization={ACM},
        date={2013},
   booktitle={Proceedings of the sixth acm international conference on web
  search and data mining},
       pages={23\ndash 32},
}

\bib{hethcote2000mathematics_SI}{article}{
      author={Hethcote, Herbert~W.},
       title={The mathematics of infectious diseases},
        date={2000},
     journal={SIAM review},
      volume={42},
      number={4},
       pages={599\ndash 653},
}

\bib{PCA:1987}{incollection}{
      author={Jolliffe, I.T.},
       title={Introduction},
        date={1986},
   booktitle={Principal component analysis},
      series={Springer Series in Statistics},
   publisher={Springer New York},
       pages={1\ndash 7},
         url={http://dx.doi.org/10.1007/978-1-4757-1904-8_1},
}

\bib{kalman1960}{article}{
      author={Kalman, Rudolph~Emil},
       title={A new approach to linear filtering and prediction problems},
        date={1960},
     journal={Journal of Fluids Engineering},
      volume={82},
      number={1},
       pages={35\ndash 45},
}

\bib{kivela2013multilayer}{article}{
      author={Kivel{\"a}, Mikko},
      author={Arenas, Alexandre},
      author={Barthelemy, Marc},
      author={Gleeson, James~P},
      author={Moreno, Yamir},
      author={Porter, Mason~A},
       title={Multilayer networks},
        date={2013},
     journal={arXiv preprint arXiv:1309.7233},
}

\bib{SIR_M1}{article}{
      author={Min, Byungjoon},
      author={Goh, K-I},
       title={Layer-crossing overhead and information spreading in multiplex
  social networks},
        date={2013},
     journal={arXiv preprint arXiv:1307.2967},
}

\bib{mohar1991laplacian}{article}{
      author={Mohar, Bojan},
      author={Alavi, Y},
      author={Chartrand, G},
      author={Oellermann, OR},
       title={The laplacian spectrum of graphs},
        date={1991},
     journal={Graph theory, combinatorics, and applications},
      volume={2},
       pages={871\ndash 898},
}

\bib{SIS}{article}{
      author={N{\aa}sell, Ingemar},
       title={The quasi-stationary distribution of the closed endemic sis
  model},
        date={1996},
     journal={Advances in Applied Probability},
       pages={895\ndash 932},
}

\bib{newman2010networks}{book}{
      author={Newman, Mark},
       title={Networks: an introduction},
   publisher={Oxford University Press},
        date={2010},
}

\bib{Jaccard}{inproceedings}{
      author={Pang-Ning, Tan},
      author={Steinbach, Michael},
      author={Kumar, Vipin},
      author={others},
       title={Introduction to data mining},
        date={2006},
   booktitle={Library of congress},
       pages={74},
}

\bib{SIR_M3}{article}{
      author={Qian, Dajun},
      author={Yagan, O.},
      author={Yang, Lei},
      author={Zhang, Junshan},
       title={Diffusion of real-time information in social-physical networks},
        date={2012Dec},
        ISSN={1930-529X},
       pages={2072\ndash 2077},
}

\bib{rodriguez2011uncovering}{article}{
      author={Rodriguez, Manuel~Gomez},
      author={Balduzzi, David},
      author={Sch{\"o}lkopf, Bernhard},
       title={Uncovering the temporal dynamics of diffusion networks},
        date={2011},
     journal={arXiv preprint arXiv:1105.0697},
}

\bib{_modern_2010}{book}{
      author={Sakurai, J.~J.},
      author={Napolitano, Jim~J.},
       title={Modern quantum mechanics},
     edition={2 edition},
   publisher={Addison-Wesley},
        ISBN={978-0-8053-8291-4},
}

\bib{SIS_M1}{article}{
      author={Sanz, Joaquin},
      author={Xia, Cheng-Yi},
      author={Meloni, Sandro},
      author={Moreno, Yamir},
       title={Dynamics of interacting diseases},
        date={2014Oct},
     journal={Phys. Rev. X},
      volume={4},
       pages={041005},
}

\bib{SIS_M2}{article}{
      author={Saumell-Mendiola, Anna},
      author={Serrano, M.~\'Angeles},
      author={Bogu\~n\'a, Mari\'an},
       title={Epidemic spreading on interconnected networks},
        date={2012Aug},
     journal={Phys. Rev. E},
      volume={86},
       pages={026106},
}

\bib{SIR}{article}{
      author={Shulgin, Boris},
      author={Stone, Lewi},
      author={Agur, Zvia},
       title={Pulse vaccination strategy in the sir epidemic model},
        date={1998},
     journal={Bulletin of Mathematical Biology},
      volume={60},
      number={6},
       pages={1123\ndash 1148},
}

\bib{sole2013spectral}{article}{
      author={Sole-Ribalta, Albert},
      author={De~Domenico, Manlio},
      author={Kouvaris, Nikos~E},
      author={Diaz-Guilera, Albert},
      author={Gomez, Sergio},
      author={Arenas, Alex},
       title={Spectral properties of the laplacian of multiplex networks},
        date={2013},
     journal={Physical Review E},
      volume={88},
      number={3},
       pages={032807},
}

\bib{ISOMAP:2000}{article}{
      author={Tenenbaum, Joshua~B.},
      author={Silva, Vin~de},
      author={Langford, John~C.},
       title={A global geometric framework for nonlinear dimensionality
  reduction},
        date={2000},
     journal={Science},
      volume={290},
      number={5500},
       pages={2319\ndash 2323},
         url={http://www.sciencemag.org/content/290/5500/2319.abstract},
}

\bib{MDS:1952}{article}{
      author={Torgerson, Warren~S.},
       title={Multidimensional scaling: I. theory and method},
        date={1952},
        ISSN={0033-3123},
     journal={Psychometrika},
      volume={17},
      number={4},
       pages={401\ndash 419},
         url={http://dx.doi.org/10.1007/BF02288916},
}

\bib{Ornstein1930}{article}{
      author={Uhlenbeck, G.~E.},
      author={Ornstein, L.~S.},
       title={On the theory of the brownian motion},
        date={1930Sep},
     journal={Phys. Rev.},
      volume={36},
       pages={823\ndash 841},
         url={http://link.aps.org/doi/10.1103/PhysRev.36.823},
}

\bib{SIS_M3}{article}{
      author={Wang, Huijuan},
      author={Li, Qian},
      author={D'Agostino, Gregorio},
      author={Havlin, Shlomo},
      author={Stanley, H.~Eugene},
      author={Van~Mieghem, Piet},
       title={Effect of the interconnected network structure on the epidemic
  threshold},
        date={2013Aug},
     journal={Phys. Rev. E},
      volume={88},
       pages={022801},
}

\bib{Wei2014JSAC}{article}{
      author={Wei, Xuetao},
      author={Valler, N.C.},
      author={Prakash, B.A.},
      author={Neamtiu, I.},
      author={Faloutsos, M.},
      author={Faloutsos, C.},
       title={Competing memes propagation on networks: A network science
  perspective},
        date={2013June},
        ISSN={0733-8716},
     journal={Selected Areas in Communications, IEEE Journal on},
      volume={31},
      number={6},
       pages={1049\ndash 1060},
}

\bib{active_passive_1}{article}{
      author={Ya\ifmmode~\breve{g}\else \u{g}\fi{}an, Osman},
      author={Gligor, Virgil},
       title={Analysis of complex contagions in random multiplex networks},
        date={2012Sep},
     journal={Phys. Rev. E},
      volume={86},
       pages={036103},
}

\bib{SIR_M4}{article}{
      author={Zhao, Dawei},
      author={Li, Lixiang},
      author={Peng, Haipeng},
      author={Luo, Qun},
      author={Yang, Yixian},
       title={Multiple routes transmitted epidemics on multiplex networks},
        date={2014},
        ISSN={0375-9601},
     journal={Physics Letters A},
      volume={378},
      number={10},
       pages={770 \ndash  776},
}

\end{biblist}
\end{bibdiv}
\addcontentsline{toc}{section}{Bibliography}

\end{document}